\documentclass[a4paper,preprintnumbers,floatfix,superscriptaddress,prl,twocolumn,showpacs]{revtex4-1}

\usepackage[utf8]{inputenc}
\usepackage[T1]{fontenc}
\usepackage[sc,osf]{mathpazo}
\usepackage{amsmath, amsthm, amssymb,amsfonts}
\usepackage{graphicx}
\usepackage{dcolumn}
\usepackage{bm}
\usepackage{bbm}
\usepackage{hyperref}
\usepackage{mathtools}
\usepackage{comment}
\usepackage{color}

\def\RR{\mathbbm{R}}
\def\CC{\mathbbm{C}}

\def\MM{\mathcal{M}}

\def\1{\mathbf{1}}
\def\0{\mathbf{0}}
\def\o{\mathbb{O}}

\def\minimize{\textrm{minimize}}
\def\st{\textrm{subject to }}
\def\maximize{\textrm{maximize}}

\def\Id{\mathbbm{1}}

\def\b{\mathbf{b}}
\def\c{\mathbf{c}}
\def\d{\mathbf{d}}
\def\e{\mathbf{e}}
\def\t{\mathbf{t}}
\def\u{\mathbf{u}}
\def\p{\mathbf{p}}
\def\q{\mathbf{q}}
\def\v{\mathbf{v}}

\def\x{\mathbf{x}}
\def\y{\mathbf{y}}
\def\z{\mathbf{z}}
\def\bxi{\bm{\xi}}
\def\bzeta{\bm{\zeta}}
\def\CC{\mathcal{C}}
\def\CD{\mathcal{D}}
\def\CP{\mathcal{P}}

\newcommand{\ket}[1]{| #1 \rangle}
\newcommand{\bra}[1]{\langle #1 |}

\newtheorem{prop}{Proposition}

\newtheorem{corollary}[prop]{Corollary}
\newtheorem{theorem}[prop]{Theorem}

\newtheorem{lemma}[prop]{Lemma}
\newtheorem{example}[prop]{Example}
\newtheorem{proposition}[prop]{Proposition}

\newcommand{\beq}{\begin{equation}}
\newcommand{\eeq}{\end{equation}}
\newcommand{\bea}[1]{\begin{equation}\begin{array}{#1}}
\newcommand{\eea}{\end{array}\end{equation}}
\newcommand{\beqn}{\begin{eqnarray}}
\newcommand{\eeqn}{\end{eqnarray}}

\renewcommand{\rho}{\varrho}

\newcommand{\processnext}[1]{%
  \ifx\listfinish#1\empty\else\listact{#1}\expandafter\processnext\fi}

\newcommand{\figref}[1]{Fig.~\ref{#1}}

\newcommand{\ea}{\end{eqnarray}}
\newcommand{\ban}{\begin{eqnarray*}}
\newcommand{\ean}{\end{eqnarray*}}

{\begin{framed}\begin{small}}
{\end{small}\end{framed}}

\DeclareGraphicsExtensions{.pdf,.png,.jpg}

\makeindex

\begin{document}
\title{A unifying framework for relaxations of the causal assumptions in Bell's theorem}
\date{\today}

\author{R.\ Chaves}
\affiliation{Institute for Physics, University of Freiburg, Rheinstrasse 10, D-79104 Freiburg, Germany}
\author{R.\ Kueng}
\affiliation{Institute for Physics, University of Freiburg, Rheinstrasse 10, D-79104 Freiburg, Germany}
\author{J.B.\ Brask}
\affiliation{D\'epartement de Physique Th\'eorique, Universit\'e de Gen\`eve, 1211 Gen\`eve, Switzerland}
\author{D.\ Gross}
\affiliation{Institute for Physics, University of Freiburg, Rheinstrasse 10, D-79104 Freiburg, Germany}

\begin{abstract}
Bell's Theorem shows that quantum mechanical correlations can violate
the constraints that the causal structure of certain experiments
impose on any classical explanation. It is thus natural to
ask to which degree the causal assumptions -- e.g.\ ``locality''
or ``measurement independence'' -- have to be relaxed in order to
allow for a classical description of such experiments.
Here, we develop a conceptual and computational framework for treating
this problem.
We employ the language of Bayesian networks to systematically construct
alternative causal structures and bound the degree of relaxation
using quantitative measures that originate from the mathematical
theory of causality. The main technical insight is that the resulting
problems can often be expressed as computationally tractable linear
programs.
We demonstrate
the versatility of the framework by applying it to a variety
of scenarios, ranging from relaxations of the measurement
independence, locality and bilocality assumptions, to a novel
causal interpretation of CHSH inequality violations.
\end{abstract}

\maketitle

The paradigmatic Bell experiment \cite{Bell1964} involves two distant
observers, each with the capability to perform one of two possible
experiments on their shares of a joint system. Bell observed that even
absent of any detailed information about the physical processes
involved, the \emph{causal structure} of the setup alone implies
strong constraints on the correlations that can arise from any
\emph{classical} description \footnote{Here enters the third
assumption in Bell's theorem, that of \emph{realism}. It states that
one can consistently  assign a value to any physical property --
independently of whether or not it is measured.
In the Bayesian network language, this is expressed by the fact that
variables are assumed to be deterministic functions of its parents.}.
The physically well-motivated causal assumptions are:
(i) \emph{measurement independence}:
experimenters can choose which
property of a system to measure, independently of how the system has
been prepared;
(ii) \emph{locality}: the results obtained by one observer cannot be
influenced by any action of the other (ideally space-like separated)
experimenter.
The resulting constraints are Bell's inequalities \cite{Bell1964}.
Quantum mechanical processes subject to the same causal structure can
violate these constraints -- a prediction that has been abundantly
verified experimentally
\cite{Freedman1972,Aspect1982,Christensen2013,Rowe2001,Giustina2013}.
This effect is commonly
referred to as
\emph{quantum non-locality}.

It is now natural to ask how stable the effect of quantum non-locality
is with respect to relaxations of the causal assumptions. Which
``degree of measurement dependence'', e.g., is required to reconcile
empirically observed correlations with a classical and local model?
Such questions are not only, we feel, of great relevance to
foundational questions -- they are also of interest to practical
applications of non-locality, e.g.\ in cryptographic protocols.
Indeed, eavesdroppers can (and do \cite{lydersen2010hacking}) exploit the failure of
a given cryptographic device to be constrained by the presumed causal
structure
to compromise its security. At the same time, it will often be
difficult to ascertain that causal assumptions hold \emph{exactly} --
which makes it important to develop a systematic quantitative theory.

Several variants of this question have recently attracted considerable
attention
\cite{Toner2003,Barrett2010,Hall2010a,Hall2010b,Hall2011,Koh2012,Manik2013,Rai2012,Biswajit2013,Thinh2013,Putz2014,Maxwell2014}.
For example, measurement dependence has been found to be a very strong
resource:
If no
restrictions are imposed on possible correlations between the
measurement choices and the source producing the particles to be
measured, any nonlocal distribution can be reproduced
\cite{Brans1988}. What is more,
only about
 about $1/15$ of a bit of
correlation between the source and measurements is sufficient
to reproduce all correlations obtained by projective measurements on a
singlet state
\cite{Barrett2010,Hall2010b,Hall2011}.
In turn, considering relaxations of the locality assumption, Toner and
Bacon showed that one bit of communication between the distant parties is again sufficient to simulate the correlations of
singlet states \cite{Toner2003}.

In this paper we provide a unifying framework for treating relaxations
of the measurement independence and locality assumptions in Bell's
theorem. To achieve this, we borrow several concepts from the
mathematical theory of \emph{causality}, a relatively young subfield
of probability theory and statistics \cite{Pearlbook,Spirtesbook}.
With the aim of describing the causal relations (rather than mere
correlations) between variables that can be extracted from empirical
observations, this community has developed a systematic and rigorous
theory of causal structures and quantitative measures of causal
influence.

Our framework rests on three observations (details are provided
below):
(i)
Alternative causal structures can systematically be represented
using the graphical notation of Bayesian networks \cite{Pearlbook}.
There, variables are associated with nodes in a graph, and directed
edges represent functional dependencies.
(ii)
These edges can be weighted by quantitative measures of causal
influence \cite{Pearlbook,Janzing2013}.
(iii)
Determining the minimum degree of influence required for a classical
explanation of observable distributions can frequently be cast as a
computationally tractable linear program.

The versatility of the framework is
demonstrated in a variety of applications. We give an operational
meaning to the violation of the CHSH inequality \cite{Clauser1969} as
the minimum amount of direct causal influence between the parties
required to reproduce the observed correlations. Considering the
Collins-Gisin scenario \cite{Collins2004}, we show that quantum
correlations are incompatible with a classical description, even if we
allow one of the parties to communicate its outcomes. We also show
that the results in \cite{Barrett2010,Hall2011} regarding
measurement-independence relaxations can be improved by considering
different Bell scenarios. Finally, we study the bilocality assumption
\cite{Branciard2010,Branciard2012} and show that although it defines a non-convex
set, its relaxation can also be cast as a linear program, naturally
quantifying the degree of non-bilocality.

\textit{Bayesian networks and measures for the relaxation of causal assumptions---}
The causal relationships between $n$ jointly distributed discrete
random variables $(X_1, \dots, X_n)$ are specified
by means of a
\emph{directed acyclic graph} (DAG).
To this end,
each variable is associated with one of the nodes of the graph.
One then says that the $X_i$'s form a \emph{Bayesian
network} with respect to the graph, if every variable can be expressed
as a deterministic function $X_i=f_i(\mathrm{PA}_i,N_i)$ of its graph-theoretic
parents $\mathrm{PA}_i$ and
an unobserved noise term $N_i$, such that the $N_i$'s are jointly
independent \footnote{We adopt the convention that
uppercase letters label random variables while
their values are denoted in lower case. For brevity, we will
sometimes suppress explicit mention of the random variables -- e.g.\
write $p(x_i, x_j)$ instead of the more precise
$p(X_i=x_i,X_j=x_j)$.}.
This is the case if and only
if the probability $p({\bf x})=p(x_1, \dots, x_n)$ is of the form
\begin{equation}
	p({\bf x}) = \prod_{i=1}^n p (x_i | \mathrm{pa}_{i} ).
\label{markov_decom}
\end{equation}
This identity encodes the causal relationships implied by the DAG \cite{Pearlbook}.

As a paradigmatic example of a DAG, consider a bipartite Bell scenario
(\figref{fig:models}a). In this scenario, two separated observers,
Alice and Bob, each perform measurements according to some inputs,
here represented by random variables $X$ and $Y$ respectively, and
obtain outcomes, represented by $A$ and $B$.
The causal model involves an explicit shared hidden variable
$\Lambda$
which mediates the correlations between $A$ and $B$.
From \eqref{markov_decom}
it follows that $p(x,y,\lambda)=p(x)p(y)p(\lambda)$
--- which reflects the measurement independence assumption.
It also follows that
$a=f_A(x,\lambda,n_A)$, $b=f_B(y,\lambda,n_B)$. We incur no loss of
generality by absorbing the local noise terms $N_A, N_B$ into
$\Lambda$ and will thus assume from now on that $a=f_A(x,\lambda),
b=f_B(y,\lambda)$
for suitable
functions $f_A, f_B$. This encodes the locality assumption.
Together, these relations
imply the well-known local hidden variable (LHV) model of Bell's theorem:
\begin{equation}
\label{LHV}
p(a,b \vert x,y) = \sum_{\lambda} p(a\vert x, \lambda) p(b\vert y, \lambda) p(\lambda).
\end{equation}

Causal mechanisms relaxing locality
(\figref{fig:models}b--d) and measurement independence
(\figref{fig:models}e) can be easily expressed using Bayesian
networks. The networks themselves, however, do not directly quantify
the degree of relaxation. Thus, one needs to devise ways of checking
and quantifying such causal dependencies. To define a sensible measure
of causal influence we introduce a core concept from the causality
literature --  \emph{interventions} \cite{Pearlbook}.

\begin{figure}[!t]
\includegraphics[width=0.98\columnwidth]{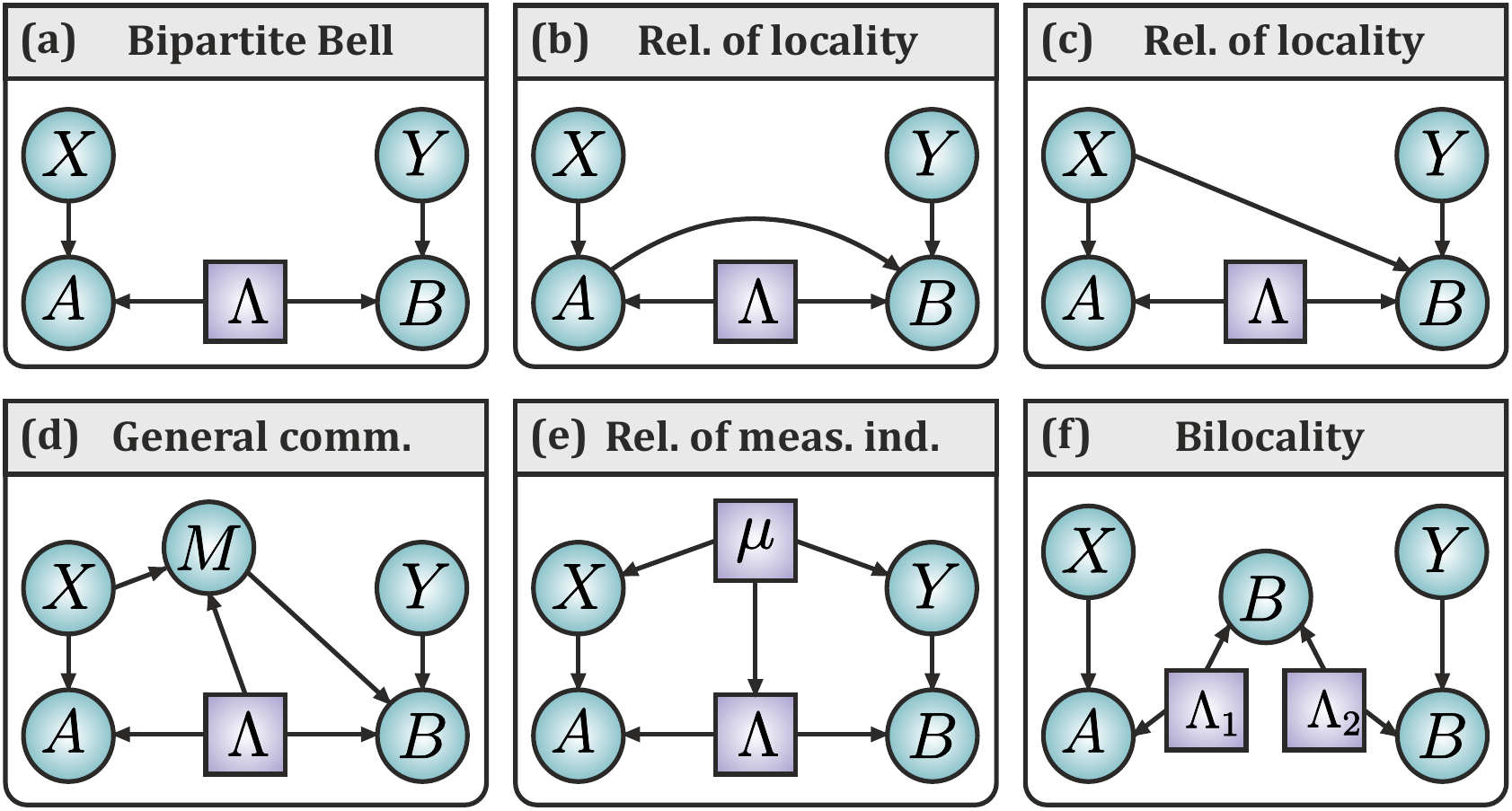}
\caption{\textbf{(a)} LHV model for the bipartite Bell scenario. \textbf{(b)} A relaxation of locality, where $A$ may have direct causal influence on $B$. \textbf{(c)} Another relaxation in which $X$ may have direct causal influence on $B$. \textbf{(d)} The most general communication scenario from Alice to Bob. \textbf{(e)} A relaxation of measurement independence, where the two inputs may be correlated, via a common ancestor, with the hidden variable $\Lambda$. \textbf{(f)} The bilocality scenario for which the two sources $\Lambda_1$ and $\Lambda_2$ are assumed to be independent. Round edges stand for observable variables while squares represent non-observable (hidden) ones.}
\label{fig:models}
\end{figure}

An intervention is the act of forcing a variable, say $X_i$, to
take on some given value $x^{\prime}_i$ and is denoted by
$do(x^{\prime}_i)$. The effect is to erase the original mechanism
$f_i(pa_i,n_i)$ and place $X_i$ under the influence of a new mechanism
that sets it to the value $x^{\prime}_i$ while keeping all other
functions $f_j$ for $j \neq i$ unperturbed. The intervention
$do(x^{\prime}_i)$ amounts to a change in the decomposition
\eqref{markov_decom}, given by \footnote{
{We note that the $do$-operation is defined only relative to a causal
model as encoded in the DAG. (The graph structure enters
(\ref{markov_trunc}) through the reference to parent nodes
$\mathrm{pa}_j$).
In particular,
$p(y \vert do(x))$ is in general different
from the usual conditional probability $p(y \vert x)$ -- these notions only
coincide if the set of parents $PA_X$ and $PA_Y$ are disjoint. For
example:  The
variables $X$ and $Y$ can be maximally correlated, i.e.\ $p(y \vert
x)\propto\delta_{x,y}$, and still $p(y \vert do(x))=p(y)$. This would occur
e.g.\ if all the correlations between the variables are mediated
via a common parent $u$, such that $p(x,y\vert u)=p(x\vert u)p(x\vert
u)$.}}
\begin{equation}
\label{markov_trunc}
	p({\bf x} \vert do(x^{\prime}_i))   =\left\{
\begin{array}{ll}
\prod_{j \neq i}^n p (x_j | \mathrm{pa}_{j} ) & \text{ if } x_i=x^{\prime}_i,\\
0 & \text{ otherwise.}%
\end{array}
\right.
\end{equation}
Considering locality relaxations, we can now define a measure $\mathcal{C}_{A \rightarrow B}$ for the \textit{direct causal influence} of $A$ into $B$ for the model in \figref{fig:models}b:
\begin{equation}
\label{meas_causal}
\mathcal{C}_{A \rightarrow B}= \sup_{b,y,a,a^{\prime}} \sum_{\lambda} p(\lambda) \vert p(b\vert do(a), y, \lambda)-p(b\vert do(a^{\prime}), y,\lambda )\vert.
\end{equation}
It is the maximum shift (averaged over the unobservable $\Lambda$) in the probability of $B$ caused by interventions in $A$.
Similarly, one can define $\mathcal{C}_{X \rightarrow B}$ for the DAG
in \figref{fig:models}c and in other situations. To highlight the
relevance of this measure,
we note that a variation of it, known as \emph{average causal effect},
can be used to quantify the effect of a drug in remedying a given
symptom \cite{Pearlbook}. We are also interested in relaxations of
measurement independence. Considering the case of a bipartite scenario (illustrated in \figref{fig:models}e and that can be easily extended to multipartite versions), we can define the measure
\begin{eqnarray}
\label{meas_corr}
\mathcal{M}_{X,Y:\lambda}  = \sum_{x,y,\lambda} \vert p(x, y,\lambda)-p(x,y)p(\lambda) \vert.
\end{eqnarray}
This can be understood as a measure of how much the inputs are correlated with the source, i.e.~how much the underlying causal model fails to comply with measurement independence.

\textit{The linear programing framework---}Given some observed
probabilities and a particular measure of relaxation, our aim is to
compute the minimum value of the measure compatible with the
observations. As sketched below, this leads to a tractable
linear program as long as there is only one unobserved variable
$\Lambda$. (However, even in case of several hidden variables, variants of these
ideas can still be used).
Details are given in the Appendix.

For simplicity we consider the usual Bell scenario of
\figref{fig:models}a.
The most general observable quantity is the joint distribution
$p(a,b,x,y)=p(a,b|x,y)p(x)p(y)$. Since we control the ``inputs'' $X$
and $Y$, their distribution carries no information and we may thus
restrict attention to $p(a,b|x,y)$.
This conditional probability is, in turn, a linear function of the
distribution of $\Lambda$.
To make this explicit, represent $p(a,b|x,y)$ as a vector $\p$ with
components $\p_j$ labeld by the multi-index $j=(a,b,x,y)$. Similarly,
identify the distribution of $\Lambda$ with a vector with components
$\q_\lambda=p(\Lambda = \lambda)$. Then from the discussion above, we
have that $\p = T \q$ where $T$ is a matrix with elements
$T_{j,\lambda} = \delta_{a,f_A(x,\lambda)}
\delta_{b,f_B(y,\lambda)}$.
Conditional expectations that include the
application  of a $do$-operation are obtained via a modified $T$ matrix.
E.g., $\q'_j=p(a,b|x,y,do(a'))=T'q$ for
$T'_{j,\lambda} = \delta_{a,a'} \delta_{b,f_B(y,\lambda)}$.
The measures $\mathcal{C}$ and $\mathcal{M}$ are easily seen to be
convex functions of the conditional probabilities $p(a,b|x,y)$ and
their variants arising from the application of $do$'s
-- and thus
convex functions of $\q$. Hence their minimization subject to the
linear constraint $T\q=\p$ for an empirically observed distribution
$\p$ is a convex optimization problem.
This remains true if only some linear function $V\p=V T\q$ (e.g.\ a
Bell inequality) of the distribution $\p$ is constrained.
The problem is not manifestly a
(computationally tractable) linear program (LP), since  neither
objective function is linear in $\q$. However, we establish in the
appendix that it can be cast as such:

\begin{theorem}\label{thm:min_max_duality}
The constrained minimization of the measures $\mathcal{C}$ and
$\mathcal{M}$ over hidden variables reproducing any observed
probability distribution can be reformulated as a primal linear
program (LP). Its solution is equivalent to
\begin{equation}
\max_{1 \leq i \leq K} \langle \v_i, V \p \rangle
, \label{eq:closed_form}
\end{equation}
where the $\left\{ \v_i \right\}_{i=1}^K$ are the vertices of the LP's
dual feasible region.
\end{theorem}
This result highlights another nice aspect of our framework.
Unlike the results in
\cite{Hall2010a,Hall2010b,Hall2011,Manik2013,Rai2012,Koh2012,Biswajit2013},
\eqref{eq:closed_form} is a closed form-expression valid for any
distribution (or observation derived from it by a linear function $V
\p$), not just the value of a specific Bell inequality. This allows
for a much more detailed description.

In the following sections, we apply our framework to a variety of applications. We focus on the results while the more technical proofs are given in the Appendices.

\textit{Novel interpretation of the CHSH inequality---}
As a first application, we show that a violation of the CHSH inequality can be interpreted as the minimal direct causal influence between the parties required to simulate the observed correlations.

Intuitively, the more nonlocal a given distribution is, the more direct causal influence between Alice and Bob should be required to simulate it. We make this intuition precise by considering the models in \figref{fig:models}b--c and the CHSH scenario (two inputs, two outputs for both Alice and Bob). For any observed distribution $p(a,b\vert x,y)$, we establish in the Appendix that
\begin{equation}
\label{CXB_CHSH}
\min \mathcal{C}_{A \rightarrow B}=\min \mathcal{C}_{X \rightarrow B}= \max \left[ 0,\mathrm{CHSH} \right] ,
\end{equation}
where the maximum should be taken over all the eight symmetries under relabelling of inputs, outputs, and parties of the CHSH quantity \cite{Clauser1969}
\begin{equation}
\label{CHSH}
\begin{split}
CHSH = \, & p(00|00) + p(00|01) + p(00|10) \\ & - p(00|11) - p^{A}(0|0) - p^{B}(0|0) ,
\end{split}
\end{equation}
where the last two terms represent the marginals for Alice and Bob respectively. The CHSH inequality stipulates that for any LHV model, $CHSH \leq 0$. Eq.~\eqref{CXB_CHSH} shows that, regardless of the particular distribution, the minimum direct causal influence is exactly quantified by the CHSH inequality violation.

Inspired by the communication scenario of Toner and Bacon \cite{Toner2003} (\figref{fig:models}d), we can also quantify the relaxation of the locality assumption as the minimum amount of communication required to simulate a given distribution. We measure the communication by the Shannon entropy $H(m)$ of the message $m$ which is sent. For a binary message, we can use our framework to prove, in complete analogy with \eqref{CXB_CHSH}, that
\begin{equation}
\min H(m)= h(\mathrm{CHSH})
\end{equation}
if $\mathrm{CHSH} > 0$ and $0$ otherwise, where $h(v)=-v\log_2v-(1-v)\log_2(1-v)$ is the binary entropy. We note that for maximal quantum violation $\mathrm{CHSH}=1/\sqrt{2} -1/2$, as produced by a single state, a message with $H(m)\approx 0.736$ bits is required. This is less than the $1$ bit of communication required by the protocol of Toner and Bacon \cite{Toner2003} for reproducing arbitrary correlations of a singlet.

\textit{Quantum nonlocality is incompatible with some locality relaxations---}
Given that violation of CHSH can be directly related to relaxation of locality, one can ask whether similar interpretations exists for other scenarios. For example, we can consider a setting with three inputs and two outputs for Alice and Bob, and consider the causal model in \figref{fig:models}b. Similar to the usual LHV model \eqref{LHV}, the correlations compatible with this model form a polytope. One facet of this polytope is
\begin{equation}
\langle E_{00} \rangle  - \langle E_{02} \rangle
- \langle E_{11} \rangle + \langle E_{12} \rangle
-\langle E_{20} \rangle + \langle E_{21} \rangle  \leq 4,
\label{eq:I3322E}
\end{equation}
where $E_{xy}=\langle A_xB_y \rangle = \sum_{a,b} (-1)^{a+b} p(a,b\vert x,y)$. This inequality can be violated by any quantum state $\ket{\psi}= \sqrt {\epsilon} \ket{00}+ \sqrt { (1-\epsilon)}\ket{11}$ with $\epsilon \neq 0,1$. Consequently, any pure entangled state -- no matter how close to separable -- generates correlations that cannot be explained even if we allow for a relaxation of the locality assumption, where one of the parties communicates its measurement outcomes to the other.

\textit{How much measurement dependence is required to causally explain nonlocal correlations?---}
The results in Refs. \cite{Barrett2010,Hall2010b,Hall2011} show that measurement dependence is a very strong resource for simulating nonlocal correlations. In fact, a mutual information as small as $I(X,Y:\lambda) \approx 0.0663$ is already sufficient to simulate all correlations obtained by (any number of) projective measurements on a single state \cite{Hall2010b,Hall2011}. Given the fundamental implication and practical relevance of increasing these requirements, we aim to find larger values for $I (X,Y: \lambda)$ by means of our framework. The result of \cite{Hall2010b,Hall2011} leaves us with three options, regarding the quantum states: either non-maximally entangled states of two qubits, two-qudit states, or states with more than two parties.

Regarding non-maximally entangled two-qubit states, we were unable to improve the minimal mutual information.
Regarding qudits, we have considered relaxations in the CGLMP scenario \cite{Collins2002} -- a bipartite scenario, where Alice and Bob each have two inputs and $d$ outcomes. The CGLMP inequality is of the form $I_{d} \leq 2$.
Assuming that a particular $I_d$-value is observed in the setting of \figref{fig:models}e, we numerically obtain the very simple relation
\begin{equation}
\label{eq.minMId}
\min \mathcal{M}=\max \left[ 0, (I_d-2)/4 \right]
\end{equation}
up to $d = 8$. Via the Pinsker inequality \cite{Fedotov2003,Hall2013}, \eqref{eq.minMId} provides a lower bound on the minimum mutual information $I(X,Y:\lambda) \geq  \mathcal{M}^2 \log_2 \mathrm{e}$. This bound implies that for any $I_d \geq 3.214$, the mutual information required exceeds the $0.0663$ obtained in Ref. \cite{Hall2011}. Using the results in Ref. \cite{Chen2006} for the scaling of the optimal quantum violation with $d$, one sees that this requires $d \geq 16$. However, we note that the bounds provided by the Pinsker inequality are usually far from tight, leaving a lot of room for improvement. Moreover -- as detailed in the Appendix -- a corresponding upper bound (obtained via the solution to the minimization of $\mathcal{M}$) is larger than the values obtained in \cite{Hall2010b,Hall2011} as soon as $d \geq 5$. Though this upper bound is not necessarily tight, we highlight the fact that for $d=2$ it gives exactly $I(X,Y:\lambda)=0.0463$, the value analytically obtained in \cite{Hall2010b,Hall2011}.

Regarding multipartite scenarios, we have considered GHZ correlations \cite{GHZ} in a tri-partite scenario where each party has two inputs and two outputs. We numerically obtain $0.090 \leq I(X,Y,Z:\lambda) \leq 0.207$. This implies that increasing the number of parties can considerably increase the measurement dependence requirements for reproducing quantum correlations.

\textit{Bilocality scenario---}
To illustrate how the formalism can also be used in generalized Bell scenarios \cite{Branciard2010,Branciard2012,Fritz2012,Chaves2012}, we briefly explore the entanglement swapping scenario~\cite{Zukowski1993} of \figref{fig:models}f (a more detailed discussion is given in the Appendix). As can be seen from the DAG, the hidden variables in this scenario are independent $ p(\lambda_1,\lambda_2)=p(\lambda_1)p(\lambda_2)$, the so-called \emph{bilocality} assumption \cite{Branciard2010,Branciard2012}.

As in Ref.~\cite{Branciard2010,Branciard2012}, we take the inputs $x,z$ and the outputs $a,c$ to be dichotomic while $b$ takes four values which we decompose in two bits as $b=(b_0,b_1)$. The distribution of hidden variables can be organized in a 64-dimensional vector ${\bf q}$ with components $q_{\alpha_0,\alpha_1,\beta_0,\beta_1,\gamma_0,\gamma_1}$, where $\alpha_x$ specifies the value of $a$ for a given $x$ (and analogously for $\gamma$, $c$ and $z$) and $\beta_{i}$ specifies the value of $b_i$. Thus together the indices label all the deterministic functions for $A$, $B$, $C$ given their parents. As shown in ~\cite{Branciard2010,Branciard2012}, the bilocality assumption is equivalent to demanding $q^{ac}_{\alpha_0,\alpha_1,\gamma_0,\gamma_1}=q^{a}_{\alpha_0,\alpha_1}q^{c}_{\gamma_0,\gamma_1}$, where $q^{ac}_{\alpha_0,\alpha_1,\gamma_0,\gamma_1}=\sum_{\beta_0,\beta_1}q_{\alpha_0,\alpha_1,\beta_0,\beta_1,\gamma_0,\gamma_1}$ is the marginal for AC etc. Similar to \eqref{meas_corr} a natural measure $\mathcal{M}_{\text{BL}}$ of non-bilocality quantifies by how much the underlying hidden variable distribution fails to comply with this constraint:
\begin{equation}
\label{eq:bilocality_measure}
\mathcal{M}_{\text{BL}}= \sum_{\alpha_0,\alpha_1,\gamma_0,\gamma_1} \vert q^{ac}_{\alpha_0,\alpha_1,\gamma_0,\gamma_1} -q^{a}_{\alpha_0,\alpha_1}q^{c}_{\gamma_0,\gamma_1} \vert.
\end{equation}
Clearly $\mathcal{M}_{\text{BL}}=0$, if and only if the bilocality constraint is fulfilled.
However, demanding bilocality imposes a quadratic constraint on the hidden variables. This results in a non-convex set which is extremely difficult characterize \cite{Branciard2010,Branciard2012,Fritz2012,Chaves2012}. Nevertheless, our framework is still useful, as using the marginals for a given observed distribution to constrain the problem further, the minimization of $\mathcal{M}_{\text{BL}}$ can be cast in terms of a linear program with a single free parameter, which is then further minimised over (see Appendix).

As an illustration we consider the non-bilocal distribution found in Refs.~\cite{Branciard2010,Branciard2012}. It can be obtained by projective measurements on a pair of identical two-qubit entangled states $\varrho=v\ket{\Psi^{-}}\bra{\Psi^{-}}+(1-v)\mathbb{I}/4$. This distribution violates the bilocality inequality $\mathcal{B}=\sqrt{|I|}+\sqrt{|J|} \leq 1$ giving a value $\mathcal{B}=\sqrt{2}v$. Using our framework we find $\mathcal{M}_{\text{BL}}= \max(2v^2-1,0)$. Thus, for this specific distribution, $\mathcal{M}_{\text{BL}}=\mathcal{B}^2-1$, so there is a one-to-one correspondence between the violation of the bilocality inequality and the minimum relaxation of the bilocality constraint required to reproduce the correlations. This assigns an operational meaning to $\mathcal{B}$.

\textit{Conclusion---}
In this work we have revisited nonlocality from a causal inference perspective and provided a linear programming framework for relaxing the measurement independence and locality assumptions in Bell's theorem. Using the framework, we have given a novel causal interpretation of violations of the CHSH inequality, and we have shown that quantum correlations are still incompatible with classical causal models even if one allows for the communication of measurement outcomes. This implies that quantum nonlocality is even stronger than previously thought. Considering a variety of scenarios, we also have shown that the results in Refs. \cite{Barrett2010,Hall2010b,Hall2011} regarding the minimal measurement dependence required to simulated nonlocal correlations can be extended. Finally we explained how the relaxation of the bilocality assumption naturally quantifies the degree of non-bilocality in an entanglement swapping experiment.

In addition to these results, we believe the generality of our framework motivates and -- more importantly --  provides a basic tool for future research. For instance, it would be interesting to understand how our framework can be generalized in order to derive useful inequalities in the context of randomness expansion, following the ideas in \cite{Koh2012}. Another natural possibility, inspired by \cite{Gallego2012,Pironio2013}, would be to look for a good measure of genuine multipartite nonlocality, by considering specific underlying signalling models. Finally, it would be interesting to understand how our treatment of the bilocality problem could be generalized and applied to the characterization of the non-convex compatibility regions of more complex quantum networks \cite{Cavalcanti2011quantum,Fritz2012,Chaves2014,Chaves2014b,Chaves2014information}.

\begin{acknowledgements}
We thank R. Luce and D. Cavalcanti for useful discussions. Research in
Freiburg is supported by the Excellence Initiative of the German
Federal and State Governments (Grant ZUK 43), the Research Innovation
Fund of the University of Freiburg, the US Army Research Office
under contracts W911NF-14-1-0098 and W911NF-14-1-0133 (Quantum
Characterization, Verification, and Validation), and the DFG. JB was
supported by the Swiss National Science Foundation (QSIT director's
reserve) and SEFRI (COST action MP1006).
\end{acknowledgements}

\bibliography{NLcausalbib}

\newpage
\section{Appendix}
For the sake of being as self-contained as possible, we start the appendix with reviewing basic concepts in convex optimization.
We then use these concepts to establish Theorem \ref{thm:min_max_duality} -- our main technical result.
As detailed below, the measures of direct causal influence \eqref{meas_causal} and measurement dependence \eqref{meas_corr}, respectively, can be recast as vector norms.
Their minimization, subject to the specific constraints of each of the causal models in Fig.\;\ref{fig:models} is then explored in detail.

\section{Review of Linear Programming}

\emph{Linear Programming} (LP) is a very powerful and widely used tool
for dealing -- both practically and theoretically --  with certain
families of convex optimization problems.
We refer to
\cite{boyd_convex_2009,barvinok_course_2002} and references therein for
an overview.
From now on we assume that vectors $\x \in \RR^n$ are represented in the standard basis $\left\{ \e_i \right\}_{i=1}^n$, i.e. $\x = \sum_{i=1}^n x_i \e_i$.
In this representation, the two vectors $\0_n := (0,\ldots,0)^T$ (the ``zero''-vector) and $\1_n := (1,\ldots,1)^T$  (the ``all-ones'' vector) will be of particular importance.
Furthermore, we are frequently going to concatenate vectors $\x \in \RR^n$ and $\y \in \RR^m$
via $\x \oplus \y := \sum_{i=1}^n x_i \e_i + \sum_{j=1}^m y_i \e_{n+j} \in \RR^{n+m}$.
Also, $\langle \cdot, \cdot \rangle$ shall denote the standard inner product of finite dimensional real vector spaces.

 There are many equivalent ways of defining the standard form of primal/dual LP's. Here we adopt the formalism of \cite{balke_probabilistic_1997}.
A convex optimization problem fits the framework of linear programming, if it can be reformulated as
\begin{eqnarray}
\gamma =  \min_{\bxi \in \RR^n} & &  \quad  \langle \c, \bxi \rangle  \label{eq:standard_primal}\\
\st & & \quad \Phi \bxi \geq \b \nonumber 	\\
& & \quad \bxi \geq \0_n, \nonumber
\end{eqnarray}
where $\c \in \RR^n$ as well as $\b \in \RR^m$ are vectors and $ \Phi: \RR^n \to \RR^m$ corresponds to an arbitrary real $m \times n$-matrix.
The inequality signs here denote generalized inequalities on $\RR^n$ and $\RR^m$, respectively.
To be concrete, two vectors $\x,\y \in \RR^n$ obey $ \y \geq \x$ if and only if $ y_i \geq x_i$ holds for all $i = 1,\ldots,n$.

It is very useful to consider linear programming problems in pairs.
An optimization of the form (\ref{eq:standard_primal}) is called a \emph{primal problem in standard form}
 and is accompanied by its \emph{dual problem (in standard form)}:
\begin{eqnarray}
\beta = \max_{\bzeta \in \RR^m} & &  \quad  \langle \bzeta, \b \rangle \label{eq:standard_dual} \\
\st & & \quad \Phi^T \bzeta \leq \c \nonumber	\\
& & \quad \bzeta \geq \0_m.	\nonumber
\end{eqnarray}
Here, $\Phi^T: \RR^m \to \RR^n$ denotes the transpose of $\Phi$ (with respect to the standard basis).
For a given pair of linear programs, we call $\bxi \in \RR^n$ \emph{primal feasible} if it obeys the constraints $\Phi \bxi \geq \b$ and $\bxi \geq \0_n$.
Likewise, we call $\bzeta \in \RR^m$ \emph{dual feasible} if $\Phi^T \bzeta \leq \c$ and $\bzeta \geq \0_m$ hold.
Furthermore, we call an LP \emph{primal feasible}, if it admits at least one primal feasible variable $\bxi$
and \emph{dual feasible}, if there exists at least one dual feasible $\bzeta$.
One crucial feature of linear programming problems is the following theorem
(see e.g. \cite[Theorem~IV.6.2 and Theorem~IV.7.2]{barvinok_course_2002})

\begin{theorem}[Weak+Strong Duality] \label{thm:duality}
Any primal feasible $\bxi$ and any dual feasible $\bzeta$ obey
\begin{equation}
\langle \c, \bxi \rangle \geq \langle \bzeta, \b \rangle \quad \textrm{(weak duality)}\label{eq:weak_duality}.
\end{equation}
Furthermore, if a given LP is either primal or dual feasible,  problems (\ref{eq:standard_primal}) and (\ref{eq:standard_dual}) are equivalent, i.e.
\begin{equation}
\gamma = \beta \quad \textrm{(strong duality)}. \label{eq:strong_duality}
\end{equation}
\end{theorem}
Strong duality is a very powerful tool, as it allows one to switch between solving primal and dual problems at will.
Moreover, the general framework of linear programming is surprsingly versatile, because
many non-linear convex optimization problems can be converted into a corresponding LP.
Here, we content ourselves with two examples which will turn out to be important for our analysis.

\begin{example}[$\ell_1$-norm calculation, \cite{boyd_convex_2009} p. 294 ] \label{ex:l1_norm}
Let $\x \in \RR^n$ be an arbitrary vector. Then
\begin{eqnarray}
\| \x\|_{\ell_1} = \min_{\t \in \RR^n} & & \quad   \langle \1_n, \t \rangle \label{eq:unconstrained_l1} \\
\st & & \quad  -\t \leq \x \leq \t. \label{eq:unconstrained_l1_constraint}
\end{eqnarray}
Note that the constraint \eqref{eq:unconstrained_l1_constraint} implicitly assures $ \t \geq \0_n$.
\end{example}

\begin{example}[$\ell_\infty$-norm calculation, \cite{boyd_convex_2009} p. 293] \label{ex:l_infty_norm}
Let $\x \in \RR^n$ be an arbitrary vector. Then
\begin{eqnarray}
\| \x \|_{\ell_\infty} &=& \min_{v \in \RR} v \\
\st & & - v \1_n \leq \x \leq v \1_n. \label{eq:LPl_infty2}
\end{eqnarray}
Note that the constraint $- v \1_n \leq \x$ is redundant if the vector of interest obeys $\x \geq \0_n$.
Also, (\ref{eq:LPl_infty2}) implicitly assures $v \geq 0$.
\end{example}

The primal LPs in examples \ref{ex:l1_norm} and \ref{ex:l_infty_norm} are not yet in standard form (\ref{eq:standard_primal}).
However, they can be converted into it by applying some straightforward reformulations -- we will come back to this later.

Another useful feature of LPs is that different minimization procedures of the above kind can be combined in order to yield an LP for a more complicated optimization problem.
An instance of  such a combination is the following result which will turn out to be crucial for our analysis.

\begin{lemma}		\label{lem:aux1}
Let $\{ \x_1,\ldots, \x_L \} \subset \RR^n$ be an arbitrary family of $L$ vectors. Then
\begin{eqnarray*}
\underset{1 \leq i \leq L}{\max} \| \x_i \|_{\ell_1} =  \underset{\substack{\t_1,\ldots,\t_L \in \RR^n \\ v \in \RR}}{\minimize} & & \quad v \\
\st & & \quad \begin{rcases*}
v \geq \langle \1_n, \t_i \rangle & \\
-\t_i \leq \x_i \leq \t_i & \\
\end{rcases*} 1 \leq i \leq L	
\end{eqnarray*}
which is a primal LP, albeit not yet in standard form.
Also, the constraints implicitly assure $\t_1,\ldots,\t_L \geq \0_n$ and $v \geq 0$.
\end{lemma}

\begin{proof}
We introduce the non-negative auxiliary vector
\begin{equation*}
\u := \sum_{i=1}^L \| \x_i \|_{\ell_1} \e_i \in \RR^L.
\end{equation*}
The equivalence
\begin{equation*}
\max_{i=1,\ldots,L} \| \x_i \|_{\ell_1} = \| \u \|_{\ell_\infty}
\end{equation*}
then follows from the definition of the $\ell_\infty$-norm.
Replacing this $\ell_\infty$-norm calculation by the corresponding LP (example \ref{ex:l_infty_norm} for non-negative vectors)
and including $L$ unconstrained $\ell_1$-norm calculations -- one for each component of $\u$ -- as ``subroutines'' (example \ref{ex:l1_norm})
yields the desired statement.
\end{proof}

Finally it is worthwhile to mention that constrained norm-minimization, e.g.
\begin{equation*}
\beta = \min_{\x \in \RR^n} \| \x \|_{\ell_1} \quad \st \quad A \x \geq \c,
\end{equation*}
can also be reformulated as a LP, because the constraint is linear.
To this end, simply include the additional linear constraint in the LP for calculating $\| \x \|_{\ell_1}$:
\begin{eqnarray}
\gamma = \min_{\x,\t\in \RR^n} & & \quad \langle \1_n, \t \rangle \label{eq:constrained_l1_norm_minimization} \\
\st & & \quad  -\t \leq \x \leq \t \nonumber \\
& & \quad  A \x \geq \c . \nonumber
\end{eqnarray}
Clearly, this is a LP.
Pushing this further, one can also handle certain types of non-linear constraints, e.g.
\begin{equation*}
\tilde{\gamma} = \min_{\x \in \RR^n} \| \x \|_{\ell_p}
\quad \st \quad
\| A \x \|_{\ell_q} \leq c
\end{equation*}
for $p,q \in \{1,\infty \}$ within the linear programming formalism.

\section{Useful results regarding LP's}
We can now use these concepts and techniques to obtain a linear programming formalism for a particular family of convex optimization problems that is relevant for our analysis. As detailed in the following two sections, the measures of direct causal influence \eqref{meas_causal} and of measurement dependence \eqref{meas_corr} can be cast as a $\ell_\infty$-norm and $\ell_1$-norm, respectively. This in turn allows us to state the associated equivalent dual problem for the minimization of each of these two measures, which is the scope of the following theorems.

\begin{theorem} \label{thm:1}
Let $A$ be a real $m \times n$-matrix,  $\left\{ M_i \right\}_{i=1}^L$ a family of $L$ real valued $k \times n$-matrices and let $\p \in \RR^m$ be an arbitrary vector.
Then, the convex optimization problem
\begin{eqnarray*}
\gamma =  \min_{\q \in \RR^n} & & \quad \max_{1 \leq i \leq L} \| M_i \q \|_{\ell_1} \label{eq:thm1} \\
\st & & \quad A \q = \p \nonumber \\
& & \quad \langle \1_n, \q \rangle = 1 \nonumber \\
& & \quad \q \geq 0 \nonumber
\end{eqnarray*}
can be reformulated as a primal LP. Its associated dual problem is given by
\begin{eqnarray*}
\underset{\substack{ \y_i \in \RR^k,\z \in \RR^m \\ w_i,u \in \RR }}{\mathrm{maximize}} & & \quad \langle \p, \z \rangle + u \\
\st & & \quad A^T \z + u \1_n \leq \sum_{i=1}^L  M_i^T \y_i \\
& & \quad - w_i \1_k \leq \y_i \leq w_i \1_k  \quad i = 1,\ldots,L \\
& & \quad \sum_{i=1}^L w_i \leq 1, \\
& & \quad w_1,\ldots,w_L \geq 0.
\end{eqnarray*}
\end{theorem}

\begin{proof}
Combining Lemma \ref{lem:aux1} -- for $\x_i = M_i \q \in \RR^k$ for $i=1,\ldots,L$ -- with the constrained minimization argument from (\ref{eq:constrained_l1_norm_minimization})
shows that the convex optimization problem (\ref{eq:thm1}) is equivalent to solving
\begin{eqnarray}
\underset{\substack{ \t_1,\ldots,\t_L \in \RR^k, \q \in \RR^n \\ v \in \RR}}{\mathrm{minimize}} & & \quad v \label{eq:thm1_aux1}\\
\st & & \quad A \q = \p \nonumber \\
& & \quad \langle \1_n, \q \rangle = 1 \nonumber\\
& & \quad \begin{rcases*}
v \geq \langle \1_k, \t_i \rangle & \nonumber\\
-\t_i \leq M_i \q \leq  \t_i & \nonumber\\
\end{rcases*} i=1,\ldots,L \nonumber\\
&& \quad \q \geq 0 \nonumber
\end{eqnarray}
which is clearly a LP. Note that the remaining optimization variables $v \in \RR$ and $\t_i \in \RR^k$ are also implicitly constrained to be non-negative.
So, in order to convert (\ref{eq:thm1_aux1}) into a primal LP in standard form  (\ref{eq:standard_primal}), we define
\begin{eqnarray*}
\bxi &:=& v \oplus \bigoplus_{i=1}^L \t_i \oplus \q,
\quad \c := 1 \bigoplus_{i=1}^L \0_k \oplus \0_n \quad \textrm{and} \\
\b &:=& (0)^{\oplus L} \oplus \left( \0_k \oplus \0_k \right)^{\oplus L} \oplus \p \oplus (-\p) \oplus 1 \oplus (-1).
\end{eqnarray*}
Counting the dimensions of the resulting vector spaces reveals $\bxi, \c \in \RR^{1+Lk + n}$ and $\b \in \RR^{L + 2 L k + 2m + 2}$.
Also, the (implicit and explicit) non-negativity constraints on $v, \t_1,\ldots,\t_L$ and $\q$ guarantee $\bxi \geq \0_{1 + Lk + n}$.
Due to our choice of $\b$, we can incorporate all relevant constraints of (\ref{eq:thm1_aux1}) in the compact expression
\begin{equation*}
\Phi \bxi \geq \b,
\end{equation*}
where $\Phi$ is the $(L + 2Lk + 2m + 2) \times (1 + Lk +n)$-matrix defined by
\begin{equation*}
\Phi =
\left(
\begin{array}{cccccc}
1 & - \1_k^T & \0_k^T & \cdots & \0_k^T & \0_n^T \\
\vdots & & & & & \vdots \\
1 & \0_k^T & \cdots & \0_k^T & - \1_k^T & \0_n^T \\
\0_k & \Id_{k \times k} & \o_{k \times k} & \cdots & \o_{k \times k} & M_1 \\
\0_k & \Id_{k \times k} & \o_{k \times k} & \cdots & \o_{k \times k} & -M_1 \\
\vdots & & & & & \vdots \\
\0_k & \o_{k \times k} & \cdots & \o_{k \times k} & \Id_{k \times k} & M_L \\
\0_k & \o_{k \times k} & \cdots & \o_{k \times k} & \Id_{k \times k} & -M_L \\
\0_m & \o_{m \times k} & \cdots & \cdots & \o_{m \times k} & A \\
\0_m & \o_{m \times k} & \cdots & \cdots & \o_{m \times k} & -A \\
0 & \0_k^T & \cdots & \cdots & \0_k^T & \1_n^T \\
0 & \0_k^T & \cdots & \cdots & \0_k^T & -\1_n^T
\end{array}
\right)
\end{equation*}
in the (extended) standard bases of the spaces $\RR^{1+Lk+n}$ and
$\RR^{L + 2Lk + 2m +2}$.
Our definitions of $\bxi, \c, \b $ and $\Phi$ now indeed convert (\ref{eq:thm1_aux1}) into primal standard form (\ref{eq:standard_primal}).
Its dual then simply corresponds to (\ref{eq:standard_dual}) which can be further simplified.
The structure of $\b$ suggests decomposing the dual variable $\bzeta \in \RR^{L + 2Lk + 2m + 2}$ into
\begin{equation}
\bzeta := \bigoplus_{i=1}^L w_i \bigoplus_{i=1}^L \left( \y_i' \oplus \y_i'' \right) \oplus \z' \oplus \z'' \oplus u' \oplus u''
\label{eq:thm1_zeta}
\end{equation}
with $w_i,u',u'' \in \RR$, $\y_i',\y_i'' \in \RR^k$ and $\z',\z''\in \RR^m$.
Using this decomposition of $\bzeta$, we obtain the following constraints from $\Phi^T \bzeta \leq \c$:
\begin{eqnarray*}
A^T (\z' - \z'' ) + \1_n (u' - u'' ) & \leq & \sum_{i=1}^L M_i \left( \y_i'' - \y_i' \right), \\
\y_i' + \y_i'' &\leq& w_i \1_k \quad \textrm{for } i=1,\ldots, L, \\
\sum_{i=1}^L w_i & \leq & 1.
\end{eqnarray*}
Also, due to $\bzeta \geq \0_{L + 2Lk + 2l+2}$, all the optimization variables are non-negative.
The objective function corresponds to
\begin{equation*}
\langle \bzeta, \b \rangle
= \langle \p, \z' - \z'' \rangle + u' - u''.
\end{equation*}
The particular form of objective function and constraints suggests to replace the non-negative variables $\z',\z'' \in \RR^m$ and $u',u'' \in \RR$ by
\begin{equation*}
\z := \z' - \z''
\quad \textrm{and} \quad
u := u' - u''
\end{equation*}
which are not constrained to be non-negative anymore.
Also, $\y_i' + \y_i'' \leq w_i \1_k$ together with $ \y_i',\y_i'' \geq 0$ implies the equivalent constraint
\begin{equation*}
- w_i \1_k \leq \y_i''- \y_i' \leq w_i \1_k
\end{equation*}
for all $1 \leq i \leq L$. This motivates to define $ \y_i := \y_i'' - \y_i'$ which is bounded by the above inequality chain, but also not constrained to be non-negative.
Putting everything together yields the desired statement
\end{proof}

\begin{theorem} \label{thm:2}
Let $A$ be a real valued $m \times n$ matrix, $\left\{ M_i \right\}_{i=1}^L$ be a family of real valued $k \times n$-matrices, $N$ a real valued $l \times n$-matrix and let $\p \in \RR^m$ as well as $c \in \RR$ be arbitrary.
The convex optimization problem
\begin{eqnarray}
\gamma = \min_{\q \in \RR^n} & & \quad \| N \q \|_{\ell_\infty} \label{eq:thm2} \label{eq:thm2_primal}\\
\st & & \quad \max_{1 \leq i \leq L} \| M_i \q \|_{\ell_1} \leq c \nonumber \nonumber\\
& & \quad A \q = \p \nonumber \nonumber\\
& & \quad \langle \1_n, \q \rangle = 1 \nonumber \\
& & \quad \q \geq 0 \nonumber
\end{eqnarray}
can be converted into a primal LP. Its associated dual LP corresponds to
\begin{eqnarray}
 \beta = \underset{\substack{\x \in \RR^l,  \y_i \in \RR^k, \z \in \RR^m \\ u,v,w_i \in \RR} }{\max} & & \quad \langle p, z \rangle + u - c v \label{eq:thm8_dual}\\
\st & & \quad A^T \z + u \1_n \leq \sum_{i=1}^L M_i^T \y_i + N^T \x \nonumber\\
& & \quad - w_i \1_k \leq \y_i \leq w_i \1_k \quad i =1 , \ldots, L \nonumber \\
& & \quad \sum_{i=1}^L w_i \leq v \nonumber \\
& & \quad \| \x \|_{\ell_1} \leq 1 \nonumber \\
& &  w_1,\ldots,w_L, v \geq 0. \nonumber
\end{eqnarray}
\end{theorem}

\begin{proof}

Proceeding along similar lines as in the previous proof one can show that (\ref{eq:thm2_primal}) is equivalent to solving
\begin{eqnarray}
\underset{\substack{\t_1,\ldots,\t_L \in \RR^k, \q \in \RR^n \\ v,\tilde{v} \in \RR}}{\minimize} & & \quad \tilde{v} \label{eq:thm2_aux1}\\
\st & & \quad - \tilde{v} \1_l \leq N \q \leq \tilde{v} \1_l \nonumber\\
& & \quad v \leq c \nonumber\\
& & \quad \begin{rcases*}
v \geq \langle \1_k, \t_i \rangle & \\
- \t_i \leq M_i \q \leq\t_i
\end{rcases*}  i = 1,\ldots,L \nonumber\\
& & \quad A \q = \p \nonumber \\
& & \quad \langle \1_n, \q \rangle = 1 \nonumber \\
& & \quad \q \geq \0_n \nonumber,
\end{eqnarray}
which is again clearly a primal LP.
Moreover, it strongly resembles the linear program (\ref{eq:thm1_aux1}).
Indeed, defining
\begin{eqnarray*}
\tilde{\c} &:=& 1 \oplus 0 \bigoplus_{i=1}^L \0_k \oplus \0_n, \\
\end{eqnarray*}
and extending $\bxi,\b$, as well as $\Phi$ from the proof of Theorem \ref{thm:1} to
\begin{eqnarray*}
\tilde{\bxi} := \tilde{v} \oplus \bxi,
\quad \tilde{\b} := \0_l \oplus \0_l \oplus (-c) \oplus \b
\end{eqnarray*}
and
\begin{equation*}
\widetilde{\Phi} =
\left(
\begin{array}{cc}
\1_{l} \oplus \1_l \oplus 0 & B \\
\0_{L+2Lk+2m + 2} & \Phi
\end{array}
\right)
,
\end{equation*}
where
\begin{equation*}
B :=
\left(
\begin{array}{cccccc}
 \0_l & \o_{l \times k} & \cdots & \cdots & \o_{l \times k} &  N \\
 \0_l & \o_{l \times k} & \cdots & \cdots & \o_{l \times k} & -N \\
 -1 & \0_k^T & \cdots & \cdots & \0_k^T & \0_n^T
\end{array}
\right)
\end{equation*}
converts (\ref{eq:thm2_aux1}) into primal standard form.
Going to the dual and simplifying it in a similar way
as shown in the previous proof
 -- decompose $\tilde{\bzeta}$ into  $\x' \oplus \x'' \oplus v \oplus \bzeta$, where $\bzeta$ was defined in (\ref{eq:thm1_zeta}) --
yields the desired statement upon noticing that
$ \langle \1_l, \x' + \x'' \rangle \leq 1$ together with $\x', \x'' \geq \0_l$
is equivalent to demanding that $\x := \x' - \x''$ obeys $\| \x \|_{\ell_1} \leq 1$,
but is not constrained to be non-negative anymore.
\end{proof}

\begin{corollary} \label{cor:3}
Suppose the $\ell_1$-norm constraint in the convex optimization (\ref{eq:thm2}) is omitted,
then the corresponding dual LP simplifies to
\begin{eqnarray}
\beta = \max_{\x \in \RR^l, \z \in \RR^m, u \in \RR} & & \quad \langle \p, \z \rangle + u \label{eq:cor3}\\
\st & & \quad A^T \z + u \1_n \leq N^T \x \nonumber \\
& & \quad \| \x \|_{\ell_1} \leq 1. \nonumber
\end{eqnarray}
If the normalization condition $\langle \1_n, \q \rangle = 1$ is dropped as well,
the optimization parameter $u$ assumes 0 and need not be considered in the dual optimization.
\end{corollary}

\begin{proof}
Omitting the $\ell_1$-norm constraint is equivalent to letting the constraint $c$ go to infinity.
Since $(-c v)$ is part of the dual's objective function (\ref{eq:thm8_dual}), this limit enforces $v = 0$.
This in turn demands $w_i = 0$ and consequently $\y_i = \0_k$ for all $i=1,\ldots,L$.
As a result, we obtain the first desired statement.

The second simplification requires a closer look at the proof of Theorem \ref{thm:2}. Doing so reveals that the constraint $\langle \1_n, \q \rangle = 1$ results
in the additional dual optimization parameter $u$. Omitting this constraint in the primal therefore implies that $u$ has to be dropped accordingly.
\end{proof}

Finally we are going to present the derivation of the second part of Theorem \ref{thm:min_max_duality}, namely that
solving an arbitrary feasible primal LP (in standard form), is equivalent to maximizing the dual problem over finitely many points -- the vertices of the dual feasible set.

\begin{proposition} \label{prop:boundedness}
\label{prop_dual}
Consider a primal feasible LP whose optimal value $\gamma$ is bounded from below.
Then this optimum is attained at one vertex $ \d_i $ of the dual feasible region
$\CD:= \left\{ \bzeta \in \RR^m: \Phi^T \zeta \leq \c, \bzeta \geq \0_m \right\}$:
\begin{equation*}
\gamma = \beta = \max_{1 \leq i \leq K} \langle \d_i, \b \rangle,
\end{equation*}
Possible unbounded directions (rays) of $\CD$ can be safely ignored.
\end{proposition}

Note that all the measures we consider -- \eqref{meas_causal}, \eqref{meas_corr} and \eqref{eq:bilocality_measure} in the main text --  are non-negative by construction.
Consequently, any reformulation of calculating (or optimizing over) these measures as a primal LP results in a bounded optimal value $\gamma \geq 0$.
Hence, Proposition \ref{prop:boundedness} is applicable, provided there is at least one hidden variable that reproduces the observed distribution, thus establishing that the LP is primal feasible.

Proposition \ref{prop:boundedness} establishes that the relevant part of the dual feasible region is bounded.
It can be deduced from duality -- Theorem \ref{thm:duality} --  and is standard.
In order to be self-contained, we provide a slightly different proof that exploits the geometry of linear programs more explicitly.

\begin{proof}[Proof of Proposition \ref{prop:boundedness}]

The fact that the primal LP is feasible and bounded assures that there is at least one dual feasible point via strong duality -- Theorem \ref{thm:duality}.
The dual feasible region $\CD$ is defined by $n+m$ linear inequalities and therefore has the structure of a convex polyhedron.
We have just established that this polyhedron is non-empty, but it is not necessarily bounded.
To see this, suppose for now that $\c \geq \0_n$ holds (this is not necessary, but will simplify our argument).
If $\Phi^T$ has a non-trivial kernel, then each element $\bar{\bzeta} \in \ker \left( \Phi^T \right) \cap \RR^m_+$
is not affected by the linear inequalities, because
\begin{equation*}
 \bar{\bzeta} \geq \0_m \quad \textrm{and} \quad  \Phi^T \bar{\bzeta} = \0_n \leq \c.
\end{equation*}
Consequently, $\CD$ contains the convex cone $\CC := \ker \left( \Phi^T \right) \cap \RR^m_+$.
Conversely, it is easy to show that the unbounded part of $\CD$ is fully contained in $\CC$.
This allows us to make a Minkowski decomposition
\begin{equation*}
\CD = \CC + \CP = \left\{ c + p: c \in C, p \in P \right\},
\end{equation*}
where $\CC$ is the unbounded conic part and $\CP$ denotes the polyhedron's remaining part.
We now aim to show that elements $\bar{\bzeta} \in \CC$ do not contribute to the actual optimization procedure and can therefore safely be ignored.
To this end, we combine the primal problem's (\ref{eq:standard_primal}) constraint $\Phi \bxi - \b \geq \0_m$ with the dual constraint $\bzeta \geq \0_m$ to obtain
$
\langle \bzeta, \b \rangle \leq \langle \bzeta, \Phi \bxi \rangle
$
for any primal feasible $\bxi \in \RR^n$.
Such a $\bxi$ is guaranteed to exist due to Theorem \ref{thm:duality} and
in particular implies for any $\bar{\bzeta} \in \CC$:
\begin{equation*}
\langle \bar{\bzeta}, \b \rangle
\leq \langle \bar{\bzeta}, \Phi \bxi \rangle
= \langle \Phi^T \bar{\bzeta}, \bxi \rangle
= 0.
\end{equation*}
Here, the last equality is due to $\bar{\bzeta} \in \ker \left( \Phi^T \right)$.
Therefore elements of $\CC$ manifestly do not contribute to the maximization and we can focus on the remaining set $\CP$.
By construction, $\CP$ is a bounded polyhedron and thus a polytope which can be characterized as the convex hull $\mathrm{conv} (\d_1,\ldots,\d_K)$ of its extremal points (Weyl-Minkowski Theorem \cite[Corollary 4.3]{barvinok_course_2002}).
However, it is a well known fact that the maximum of a linear (or more generally: any concave) function over a convex polytope is attained at one of its extreme pointes, i.e. vertices.
\end{proof}

\section{Relaxation of Locality}

In this section we will analyze the relaxation of the locality assumption, as exemplified by the DAGs depicted in \figref{fig:models}b--d.
In particular, we will show that evaluating the minimal direct causal influence -- see equation \eqref{meas_causal} in the main text -- that is required to simulate a given non-local distribution can be recast as a LP. Consequently, it can be determined efficiently for any observed probability distribution.

We begin analyzing in details the scenario depicted in \figref{fig:models}c. There, the input $X$ of Alice has a direct causal influence over the outcome $B$ of Bob. We consider the general, finite case where Alice has $m_x$ inputs and $o_a$ outputs, that is, $x=0,\dots,m_x-1$ and $a=0,\dots,o_a-1$ (and analogously for Bob). Variations of this signalling model can be easily constructed and will be briefly discussed at the end of this section.

The signalling model in \figref{fig:models}c requires a hidden variable $\lambda$ assuming $n=o_{a}^{m_x} o_{b}^{m_x m_y}$ possible values.
The causal structure assures $a = f_A (x, \lambda)$ which resembles the LHV model (\figref{fig:models}a).
This is not the case for $b$, which can depend on $x$,$y$ and $\lambda$ -- i.e. $b=f_B(x,y,\lambda)$.
Consequently there are $o_{a}^{m_x}$ possible deterministic functions $f_A$ and $o_{b}^{m_x m_y}$ possible deterministic functions $f_B$.
In turn, we can split up the hidden variable into $\lambda=(\lambda_a,\lambda_b)=(\alpha_0,\dots,\alpha_{m_x-1},\beta_{0,0},\beta_{0,1},\dots,\beta_{m_x-1,m_y-1})$ where $\alpha_x=0,\dots,o_a-1 $ determines the value of $a$ given $x$. Similarly, $\beta_{x,y}= 0,\dots, o_b -1 $ specifies the value of $b$ given $x$ and $y$.
Following \eqref{markov_decom} the observed distribution can be decomposed in the following way:
\begin{equation}
p(a,b\vert x,y)= \sum_{\lambda} p(a\vert x,\lambda) p(b\vert x,y,\lambda)p(\lambda).
\end{equation}

Given such a signalling model and some observed constraints, our task is to find the minimum value of $\mathcal{C}_{X \rightarrow B}$.
Similarly to \eqref{meas_causal}, this quantity can be defined as
\begin{equation}
\label{meas_causal_app}
\mathcal{C}_{X \rightarrow B}= \sup_{b,y,x,x^{\prime}} \sum_{\lambda}  p(\lambda)\vert p(b\vert do(x), y, \lambda)-p(b\vert do(x^{\prime}), y,\lambda)\vert,
\end{equation}
which quantifies the amount of signalling required to explain the observation. Moving on, we note that
\begin{eqnarray}
\nonumber
& &\sum_{\lambda}  p(\lambda)\vert p(b \vert do(x), y, \lambda)-p(b\vert do(x^{\prime}),y,\lambda)\vert \\ \label{eq:aux1}
& &= \sum_{\lambda}  p(\lambda) \vert \delta_{b,f_B(x,y,\lambda)} - \delta_{b,f_B(x^{\prime},y,\lambda)}\vert \\ \nonumber
& &=\sum_{i} q_i v_i= \langle \v, \q \rangle,
\end{eqnarray}
where we have identified $p(\lambda)$ with the $n$-dimensional vector $ {\bf q}$ via $\langle \e_i, \q \rangle = p (\lambda_i)$. The vector ${\bf v}={\bf v}(x,x^{\prime},y,b)$ only consists of $1$'s and $0$'s
and fully characterizes the action of the Kronecker-symbols in (\ref{eq:aux1}).
By doing so, the measure of causal influence \eqref{meas_causal_app} can be recast as
\begin{equation}
\label{Msign_app}
\mathcal{C}_{X \rightarrow B} = \max_{i=1,\dots,L} \langle {\bf q} , {\bf v}_i \rangle = \| C {\bf q} \|_{\infty}.
\end{equation}
Here, the index $i$ parametrizes one of the $L$ possible instances of $(x,x^{\prime},y,b)$ with $x \neq x^{\prime}$ and $\v_i = v(x,x',y,b)$ denotes the vector corresponding to that instance.
The last equality in (\ref{Msign_app}) then follows from introducing $C := \sum_{i=1}^L | \e_i \rangle \langle \v_i |$ and the definition of the $\ell_\infty$-norm.
Consequently, minimizing $\mathcal{C}_{X \to B}$ over all hidden variables that are compatible with our observations is equivalent to solving
\begin{eqnarray}
\label{eq:CLP1}
\underset{ {\bf q} \in \RR^{n}}{\minimize} & & \quad \| C {\bf q} \|_{\infty} \label{min_sign}  \\ \nonumber
\st & & \quad V T{\bf q} = V \p  \\
& & \quad \langle \1_n, {\bf q} \rangle = 1 	\label{eq:CLP2} \\ \nonumber
& & \quad {\bf q}  \geq \0_n.
\end{eqnarray}
Corollary \ref{cor:3} assures that this optimization problem can be translated into a LP in standard form.
As already mentioned in the main text, $V \p$ denotes the vector representing the correlations under consideration -- the probability distribution itself ($V = \Id$) or a function of it, e.g., a Bell inequality ( $V = | \e_1 \rangle \langle \b |$ for some $b \in \RR^m$) --  and the matrix $VT$ maps the underlying hidden variable states to the actually observed vector $V \p$ .

Given any observed distribution $V \p$ of interest, one can easily implement this linear program and solve it efficiently.
However, we are also interested in deriving an analytical solution which is valid for any vector $\p$ encoding the full probability distribution $p(a,b \vert x,y)$.
Subjecting to the full probability distribution $\p$ in particular guarantees that the normalization constraint \eqref{eq:CLP2}
is already assured by $T \q = \p$. This allows for dropping this constraint without loss of generality.
Proposition \ref{prop_dual} serves precisely the purpose of obtaining such an analytical expression, as it --  in combination with Corollary \ref{cor:3} --  assures that solving (\ref{eq:CLP1}) is equivalent to evaluating
\begin{equation*}
\max_{1 \leq i \leq K} \langle \d_i, V \p \rangle,
\end{equation*}
where $\left\{ {\bf d}_i \right\}_{i=1}^K$ denotes the vertices of the dual feasible region in \eqref{eq:cor3}.
Standard algorithms like PORTA \cite{porta} allow for evaluating these extremal points.
We have performed such an analysis for the particular case of the CHSH scenario ($m_x=m_y=o_a=o_b=2$).
We list all the $13$ vertices of the LP's dual feasible region in Table \ref{extremal_points_CHSH}.
Nicely, we see that all the extremal points can be divided into three types:
 i) the trivial vector $\0_m$, ii) the symmetries of the CHSH inequality vector, for example
\begin{equation}
\label{CHSH_app}
p^{AB}_{00\vert 00} +p^{AB}_{00\vert 01} +p^{AB}_{00\vert 10} -p^{AB}_{00\vert 11} -p^{A}_{0\vert 0} -p^{B}_{0\vert 0}
\end{equation}
and iii) the non-signalling conditions, for instance
\begin{equation}
\label{NS_app}
-p^{AB}_{01\vert 00} - p^{AB}_{11\vert 00} +  p^{AB}_{01\vert 10} + p^{AB}_{11\vert 10}.
\end{equation}
Here, we have used the short hand notation $p^{AB}_{ab\vert xy}=p(a,b \vert x,y)$ and similarly for the marginals.

For any non-signalling distribution, the conditions of the third type vanish and the corresponding vertices need not be considered.
Therefore we arrive at the result stated in the main text, namely
\begin{equation*}
\min \mathcal{C}_{X \rightarrow B}= \max \left[ 0,\mathrm{CHSH} \right],
\end{equation*}
where the maximum is taken over all the eight symmetries of the $\mathrm{CHSH}$ inequality.

\begin{table*}
\begin{tabular}{|c| c c c c c c c c c c c c c c c c|} \hline
\multicolumn{17}{|c|}{List of extremal points}\\
\hline
\textbf{\#}
&$p^{00}_{00}$&$p^{01}_{00}$&$p^{00}_{10}$&$p^{00}_{11}$ &$p^{01}_{00}$&$p^{01}_{01}$&$p^{01}_{10}$&$p^{01}_{11}$ &$p^{10}_{00}$&$p^{10}_{01}$&$p^{10}_{10}$&$p^{10}_{11}$ &$p^{11}_{00}$&$p^{11}_{01}$&$p^{11}_{10}$&$p^{11}_{11}$
\\

\hline
1
&0&0&0&0&0&0&0&0&0&0&0&0&0&0&0&0
\\

\hline
2
&0 &-1/2 &0 &1/2 &0 &-1/2 &-1 &-1/2 &0 &-1/2 &0 &1/2 &0 &1/2 &0 &-1/2
\\

\hline
3
&0 &-1/2 &0  &1/2 &0 &-1/2 &-1 &-1/2 &0  &1/2 &0 &-1/2  &0 &-1/2  &0  &1/2
\\

\hline
4
&0 &-1/2 &0  &1/2 &0  &1/2  &0 &-1/2 &0 &-1/2 &0  &1/2  &0 &-1/2 &-1 &-1/2
\\

\hline
5
&0 &-1/2 &0  &1/2 &0  &1/2  &0 &-1/2 &0  &1/2 &0 &-1/2 &-1 &-1/2  &0 &-1/2
\\

\hline
6
&0  &1/2 &0 &-1/2 &0 &-1/2  &0  &1/2 &0 &-1/2 &0  &1/2  &0 &-1/2 -&1 &-1/2
\\

\hline
7
&0  &1/2 &0 &-1/2 &0 &-1/2  &0  &1/2 &0  &1/2 &0 &-1/2 &-1 &-1/2  &0 &-1/2
\\

\hline
8
&0  &1/2 &0 &-1/2 &0  &1/2  &1  &1/2 &0 &-1/2 &0  &1/2 &-1 &-1/2 &-1 &-3/2
\\

\hline
9
&0  &1/2 &0 &-1/2 &0  &1/2  &1  &1/2 &0  &1/2 &0 &-1/2 &-1 &-3/2 &-1 &-1/2
\\

\hline
10
&0   &-1 &0   &-1 &0    &0  &0    &0 &0    &1 &0    &1  &0    &0  &0    &0
\\

\hline
11
&0    &0 &0    &0 &0   &-1  &0   &-1 &0    &0 &0    &0  &0    &1  &0    &1
\\

\hline
12
&0    &0 &0    &0 &0    &1  &0    &1 &0    &0 &0    &0  &0   &-1  &0   &-1
\\

\hline
13
&0    &1 &0    &1 &0    &0  &0    &0 &0   &-1 &0   &-1  &0    &0  &0    &0
\\

\hline
\end{tabular}
\caption{Extremal points for the feasible region in the dual problem \eqref{eq:cor3} associated with the CHSH scenario. In the notation above, $p^{xy}_{ab}$ corresponds to $p(a,b \vert x,y)$. The extremal points 2-9 can be easily seen to correspond to the symmetries of the CHSH inequality. Take for instance point $2$ which can be written as the CHSH operator in \eqref{CHSH_app}. The extremal points 10-13 correspond to the non-signalling conditions. For instance, point $10$ corresponds to \eqref{NS_app} and is zero for any non-signalling distribution.
} \label{extremal_points_CHSH}
\end{table*}

Having such a causal interpretation of the CHSH inequality at hand, one can wonder the same holds true for other Bell inequalities, for instance the $(I_{3322} \leq 0)$-inequality \cite{Collins2004} (three inputs for Alice and Bob with two outcomes each).
Dwelling on the model in \figref{fig:models}c we show that the $I_{3322}$ inequality only provides a lower bound to the actual value of $\mathcal{C}_{X \rightarrow B}$
required to simulate a given nonlocal distribution. This is illustrated in Fig.\;\ref{fig:I3322plot}.
To be more concrete, we consider the particular full probability distribution
\begin{equation}
p(a,b \vert x,y)= v p_{\text{PR}}+(1-v) p_{\text{W}},
\label{eq:I3322_distribution}
\end{equation}
where
\begin{equation*}
p_{\text{PR}}\left(  a,b | x,y \right)  =\left\{
\begin{array}{ll}
1/2 & \text{if } a+ b=1\mod 2 \text{, } x+y=3,   \\
1/2 & \text{if } a+ b=0 \mod 2 \text{, } x+y \neq 3,   \\
0 & \text{otherwise,}%
\end{array}
\right.
\end{equation*}
denotes the generalization of the PR box maximally violating the $I_{3322}$-inequality (achieving $I_{3322} = 1$) and
\begin{equation*}
p_{\text{W}}\left(  a,b | x,y \right)  =1/4
\end{equation*}
denotes the uniform distribution (achieving $I_{3322} = -1$).
Such a full probability distribution results in $I_{3322}=2v-1$.
We numerically see that
\begin{equation*}
\mathcal{C}_{X \rightarrow B}=\max \left[0,(2v-1)/2\right]=\max \left[0,I_{3322}/2\right]
\end{equation*}
holds, if we take into account the full probability distribution. However, if we instead only impose a fixed value of the $I_{3322}$-inequality (plus nonsignalling and normalization constraints) we
numerically (see Fig.\;\ref{fig:I3322plot}) arrive at
\begin{equation*}
\min \mathcal{C}_{X \rightarrow B}
=
\left\{
\begin{array}{ll}
0 & \text{for } I_{3322} \leq 0, \\
(2/5)*I_{3322} & \text{for } 0 \leq I_{3322} \leq 0.714,   \\
(1/4)*(3I_{3322}-1)& \text{for } 0.714 \leq I_{3322} \leq 1.
\end{array}
\right.
\end{equation*}
This shows that different distributions achieving the same value for $I_{3322}$ may have quite different requirements in order to be simulated. Moreover, this result highlights another nice aspect of our framework. Unlike the results in \cite{Hall2010a,Hall2010b,Hall2011,Manik2013,Rai2012,Koh2012,Biswajit2013}, it can take into account the full probability distribution, not just the value of a specific Bell inequality. This allows for a much more accurate description.

\begin{figure}[!t]
\includegraphics[width=0.9\columnwidth]{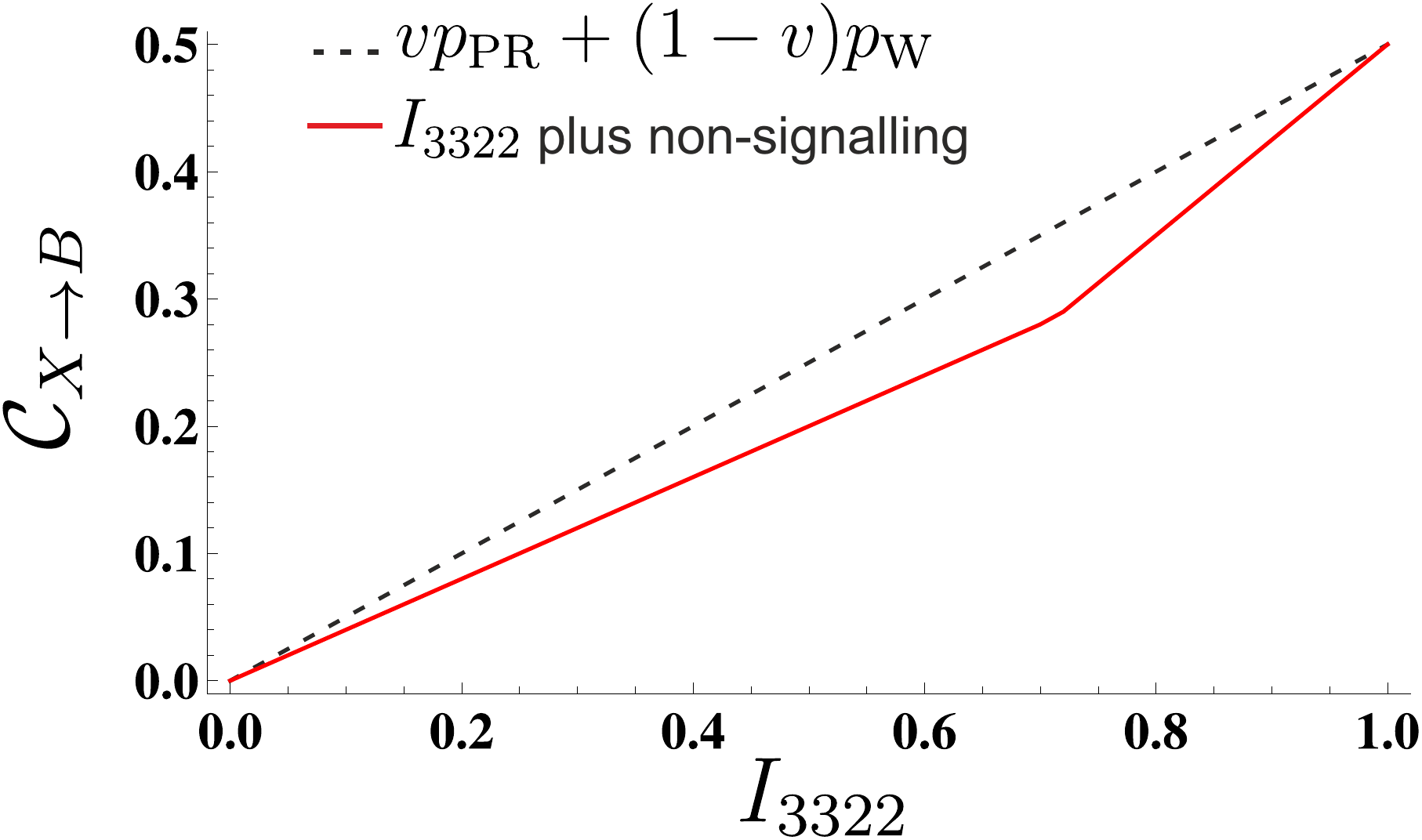}
\caption{The value of $\min \mathcal{C}_{X \rightarrow B}$ as function of the $I_{3322}$ value. The black curve represents the case where the full probability distribution defined in \eqref{eq:I3322_distribution} is taken into account. The red curve is obtained by minimizing $\mathcal{C}_{X \rightarrow B}$ for a given value of $I_{3322}$ subject to non-signalling and normalization constraints.
}
\label{fig:I3322plot}
\end{figure}

An almost identical analysis can be done for the model displayed in \figref{fig:models}b. Using \eqref{markov_decom}, the observed distribution can be decomposed as:
\begin{equation}
\label{modelAB}
p(a,b\vert x,y)= \sum_{\lambda} p(a\vert x,\lambda)  p(b\vert a,y,\lambda).
\end{equation}
Using the measure of direct causal influence \eqref{meas_causal} for $\mathcal{C}_{A \rightarrow B}$, revisiting the CHSH scenario, we can once more conclude
\begin{equation}
\label{CAB_CHSH_app}
\min \mathcal{C}_{A \rightarrow B}= \max \left[ 0,\mathrm{CHSH} \right].
\end{equation}
In particular, this implies that such a model -- where one of the parties communicates its outcomes -- is capable of simulating any nonlocal distributions in the CHSH scenario.

Interestingly, things change drastically if we move on to the $I_{3322}$ scenario.
It is worthwhile to point out that model \eqref{modelAB} restricts the hidden variables to a region characterized by finitely many inequalities.
Therefore, analogously to the usual LHV model \eqref{LHV}, the feasible region is a polytope.
Using the software PORTA we found different classes of non-trivial inequalities that define the compatibility region of this model.
As shown in the main text -- equation (\ref{eq:I3322E}) -- one of these inequalities corresponds to
\begin{equation*}
I_{A\rightarrow B}=\langle E_{00} \rangle  - \langle E_{02} \rangle
- \langle E_{11} \rangle + \langle E_{12} \rangle
-\langle E_{20} \rangle + \langle E_{21} \rangle  \leq 4,
\end{equation*}
where $E_{xy}=\langle A_xB_y \rangle = \sum_{a,b} (-1)^{a+b} p(a,b\vert x,y)$.
We now show that this inequality can be violated by any quantum state $\ket{\psi}= \sqrt {\epsilon} \ket{00}+ \sqrt { (1-\epsilon)}\ket{11}$ with $\epsilon \neq 0,1$.
To arrive at such a statement, it suffices to consider that Alice and Bob perform projective measurements on the X-Z plane of the Bloch sphere.
More concretely, Alice measures observables of the form $O^{A}_{x}= \cos (\theta^{A}_{x}) Z + \sin (\theta^{A}_{x}) X$
and so does Bob whose observables we denote by $O^{B}_{x}$.
Here, $X$ and $Z$ refer to the Pauli matrices. For such particular measurements, the correlators $E_{xy}=\langle A_xB_y \rangle$ simply correspond to
\begin{equation*}
E_{xy}=\cos (\theta^{A}_{x}) \cos (\theta^{B}_{y}) +
 2 \sqrt {\epsilon(1- \epsilon)} \sin (\theta^{A}_{x}) \sin (\theta^{B}_{y}).
\end{equation*}
Choosing the angles such that $\theta^{A}_{0}=0$, $\theta^{A}_{1}=\pi$, $\theta^{A}_{2}=\pi/2$, $\theta^{B}_{0}=0$ and $\theta^{B}_{2}=-\pi$ we obtain
\begin{equation*}
I_{A\rightarrow B}= 3 + \cos (\theta^{B}_{1}) + 2 \sqrt {\epsilon(1- \epsilon)} \sin (\theta^{B}_{1}).
\end{equation*}
This expression exceeds $4$ for any $\epsilon \neq 0,1$, provided that we choose $\theta^{B}_{1}$ sufficiently small compared to $2 \sqrt {\epsilon(1- \epsilon)}$. This result shows that even relaxing some of assumptions in Bell's theorem -- in this particular case, the fact that Alice outcomes cannot have a direct causal influence over Bob outcomes -- may not be enough to causally explain quantum correlations.

A similar analysis can be performed for the communication model of \figref{fig:models}d.
Such a model implies the following decomposition of the distribution observed:
\begin{equation*}
p(a,b\vert x,y)= \sum_{\lambda,m} p(a\vert x,\lambda) p(m \vert x,a,\lambda) p(b\vert m,y,\lambda)p(\lambda).
\end{equation*}
Such an expression suggests to decompose the hidden variable into $\lambda=(\lambda_{\alpha},\lambda_{\beta},\lambda_{m})$.
By doing so, one can perform an analysis similar to the one above and define a measure of causal influence similar to \eqref{meas_causal}. However, inspired by the communication model of Toner and Bacon \cite{Toner2003}, we directly proceed to analyzing the amount of communication between Alice and Bob required to classically reproduce the distribution observed.
We quantify the information content of a binary message $m$ sent from Alice to Bob via its Shannon entropy $H (m)$.
Due to the highly non-linear character of entropies, the optimizations involving $H(m)$ are quite hard in general.
Fortunately in the particular case of binary messages, minimizing $H (m)$ is equivalent to minimizing
\begin{eqnarray}
p(m=0) & & =\sum_{a,x,\lambda} p(m=0 \vert x, a, \lambda) p(a \vert x, \lambda) p(x) p(\lambda) \\ \nonumber
& & =(1/m_x) \sum_{a,\lambda} p(m=0 \vert x, a, \lambda) p(a \vert x, \lambda)  p(\lambda) \\ \nonumber
& & = \langle \v,\q\rangle.
\end{eqnarray}
Here, we have once more identified $p(\lambda)$ with the vector $\q$ and the components of $\v$ correspond to $v_i = \sum_a p (m=0|x,a,\lambda_i) p (a| x, \lambda_i)$.
Also, we have without loss of generality considered a uniform distribution of Alice's inputs -- i.e. $p (x) = 1/m_x$ -- in the second line.
Consequently, the constrained minimization of $p(m=0)$ (and thus $H(m)$) simply corresponds to
\begin{eqnarray*}
\underset{\q \in \RR^n}{\minimize} & & \quad \langle \v, \q \rangle \\
\st & & \quad T \q = \p \\
& & \quad \langle \1_n , \q \rangle = 1 \\
& & \quad \q \geq \0_n,
\end{eqnarray*}
which is clearly a primal LP.
Computing the extremal points of the dual problem allows us to infer a novel relation between the degree of nonlocality and the minimum communication required to simulate it.
Namely, $\min p(m=0)= \max \left[ 0, \mathrm{CHSH}_{\Pi} \right]$ which in turn implies
\begin{equation*}
\min H (m)
=
\left\{
\begin{array}{ll}
h \left( \mathrm{CHSH} \right) & \text{for } \mathrm{CHSH} \geq 0, \\
0& \text{else.}   \\
\end{array}
\right.
\end{equation*}
Here, $h$ denotes the binary entropy given by $h(v)=-v\log_2 v -(1-v) \log_2 (1-v)$.

These results on the relaxation of the locality assumption, in addition to fundamental implications and relevance in nonlocal protocols, can also be used to compute the minimum causal influences/communication required to causally explain the nonlocal correlations observed in experimental realizations of Bell's tests where the space-like separation is not achieved \cite{Rowe2001,Giustina2013}.

\section{Measurement dependence models}

In this section we focus on the measure $\MM_{X,Y:\lambda}$ -- see equation \eqref{meas_corr} in the main text -- which quantifies the degree of measurement dependence in a given causal model. Similar to the previous section, we are going to show that determining the minimal degree of measurement dependence required to reproduce a given non-local distribution can be done via solving a LP.

To illustrate this, we consider the simplest scenario of measurement dependence in detail. Such a model is displayed in \figref{fig:models}e
and involves a bipartite Bell scenario, where the measurement inputs $X$ of Alice and $Y$ of Bob, respectively, can be correlated with the source $\Lambda$ producing the particles to be measured.

Without loss of generality, we model such correlations by introducing an additional hidden variable $\mu$ which serves as a common ancestor for $x$, $y$ and $\lambda$.
This suggests to decompose this common ancestor into $\mu=(\mu_x, \mu_y, \mu_{\lambda})$.
We can assume $x=\mu_x$, $y=\mu_y$ and $\lambda=\mu_{\lambda}$ without loss of generality ($x$, $y$ and $\lambda$ are deterministic functions of their common ancestor $\mu$).
If Alice's apparatus has $m_x$ inputs (i.e. $x = 0,\ldots,m_x-1$) and $o_a$ outputs (i.e. $a = 0,\ldots,o_a-1$),
and similarly for Bob, $n = m_x m_y o_a^{m_x} o_a^{m_y}$ different instances of $\mu$ suffice to fully characterize the common ancestor's influence.
Similar to the previous section, we can use this discrete nature of $\mu$ to identify any probability distribution $p (\mu): \Xi \to [0,1]$ uniquely with a non-negative, real vector $ \q$ via
\begin{equation}
q_i = \langle {\bf e}_i, {\bf q} \rangle = p (\mu_i) \quad i=1,\ldots,n.
\end{equation}
Likewise, we can rewrite the observed probability distribution $p (a,b | x,y )$ as
\begin{eqnarray*}
& & p(a,b\vert x,y) \\
& &= \frac{1}{p(x,y)} \sum_{\mu, \lambda} p(a \vert x, \lambda)p(b \vert y, \lambda) p(x \vert \mu)p(y \vert \mu)p(\lambda \vert \mu) p(\mu) \\ \nonumber
& &= \frac{1}{p(x,y)} \sum_{\mu_{\lambda}} p(a \vert x,\mu_{\lambda}) p(b \vert y , \mu_{\lambda}) p(\mu_{\lambda})  \\ \nonumber
& & = \langle {\bf v} (x,y,a,b,\lambda), {\bf q} \rangle.
\end{eqnarray*}
The usefulness of such vectorial identifications becomes apparent when taking a closer look at the measure of correlation \eqref{meas_corr}.
Indeed,
\begin{eqnarray}
\label{meas_corr_app}
\mathcal{M} & & = \sum_{x,y,\lambda} \vert p(x, y,\lambda)-p(x,y)p(\lambda) \vert \\ \nonumber
& &= \sum_{x,y,\lambda} \vert \sum_{\mu} \delta_{\lambda,\mu_{\lambda}}(\delta_{x,\mu_x}\delta_{y,\mu_y}-p(x,y))p(\mu) \vert \\ \nonumber
& &= \sum_{x,y,\lambda} \vert \langle {\bf v}(x,y,\lambda), {\bf q} \rangle \vert \\ \nonumber
& &=  \| M {\bf q} \|_{\ell_1},
\end{eqnarray}
where $M $ denotes  the real $k \times n $ matrix $M = \sum_{j=1}^{k} |{\bf e}_j \rangle \langle  {\bf v}(x,y,\lambda)|$.
Note that this matrix implicitly depends on $p(x,y)$. However, $p (x,y)$ is an observable quantity and thus available.
Moreover, one is typically interested in the case, where said distribution for the inputs is uniformly distributed -- i.e. $p(x,y)=1/(m_x m_y)$.

It is worthwhile to point out that different measures of measurement dependence have been considered in the literature. For instance, in Ref. \cite{Hall2010b} the following measure of correlation has been proposed:
\begin{equation*}
\mathcal{M}_{\text{Hall}}  =  \sup_{x,x^{\prime},y,y^{\prime}}\sum_{y} \vert p(\lambda\vert x,y)-p(\lambda\vert x^{\prime},y^{\prime}) \vert.
\end{equation*}
Similarly to \eqref{meas_corr_app}, we can rewrite this measure as a $\ell_1$-norm, namely
\begin{equation*}
\mathcal{M}_{\text{Hall}}  =  \max_{i=1,\ldots,L} \| M_i {\bf q} \|_{\ell_1}.
\end{equation*}
The constrained minimization of both $\mathcal{M}$ and $\mathcal{M}_{\text{Hall}}$ consequently corresponds to the following optimization:
\begin{eqnarray}
\underset{{\bf q} \in \RR^n}{\minimize} & & \quad \max_{i=1,\ldots,L} \| M_i {\bf q} \|_{\ell_1} 	\label{eq:convex_min_max_problem_app} \\
\st & & \quad V {\bf q} =  {\bf \tilde{p}} \nonumber \\
& & \quad \langle \1_n, {\bf q} \rangle = 1 \nonumber	\\
& & \quad {\bf q} \geq \0_n,\nonumber
\end{eqnarray}
Theorem \ref{thm:1} assures that such an optimization can be recast as a primal LP in standard form.

In this work we have opted to focus on the measure defined in \eqref{meas_corr_app}. The reason for that is two-fold.
Firstly, such a choice assures $L=1$ and numerically solving the corresponding  LP is substantially faster.
The second reason stems from the fact that \eqref{meas_corr_app} is proportional to the variational distance between the distributions $p(x,y,\lambda)$ and $p(x,y)p(\lambda)$.
Knowledge of the total variational distance allows to lower-bound the mutual information between $(X,Y)$ and $\Lambda$ via the Pinsker inequality  \cite{Fedotov2003,Hall2013}:
\begin{equation*}
I(X,Y:\Lambda) \geq  \mathcal{M}^2 \log_2 \mathrm{e}.
\end{equation*}
A converse bound on $I (X,Y:\Lambda)$ is obtained by noting that the (linear program) solution to the minimization of $\mathcal{M}$ returns a specific hidden variable model, for which we can readily compute the mutual information.

Using measure \eqref{meas_corr_app}, we have considered many different Bell scenarios.
This was already mentioned in the main text.
In particular we refer to Fig.\;\ref{fig:CGLMPplot} where we consider the CGLMP scenario \cite{Collins2002} -- a bipartite model, where Alice and Bob measure one out of two observables each of them having $d$ possible outcomes.
The corresponding CGLMP inequality is of the form $I_{d} \leq 2$, where the local bound of $2$ and the maximal violation of $4$ are independent of the number of possible outcomes $d$.
Imposing the value of the $I_d$ inequality ad imposing non-signalling and the normalization constraints we numerically obtain
a very simple relation up to $d = 8$, namely
\begin{equation*}
\min \mathcal{M}=\max \left[ 0, (I_d-2)/4 \right].
\end{equation*}
Conversely, we have also considered specific quantum realizations. For $d=2,5,7$ we have numerically optimized over quantum states and projective measurements maximizing the corresponding $I_d$ inequality.
With the resulting quantum probability distribution at hand, we computed $\mathcal{M}$ and inferred lower and upper bounds for $I (X,Y: \Lambda)$ in turn.
These results are depicted in  Fig.\;\ref{fig:CGLMPplot} and we refer to the corresponding section in the main text for further insights concerning measurement dependence.

\begin{figure}[!t]
\includegraphics[width=0.9\columnwidth]{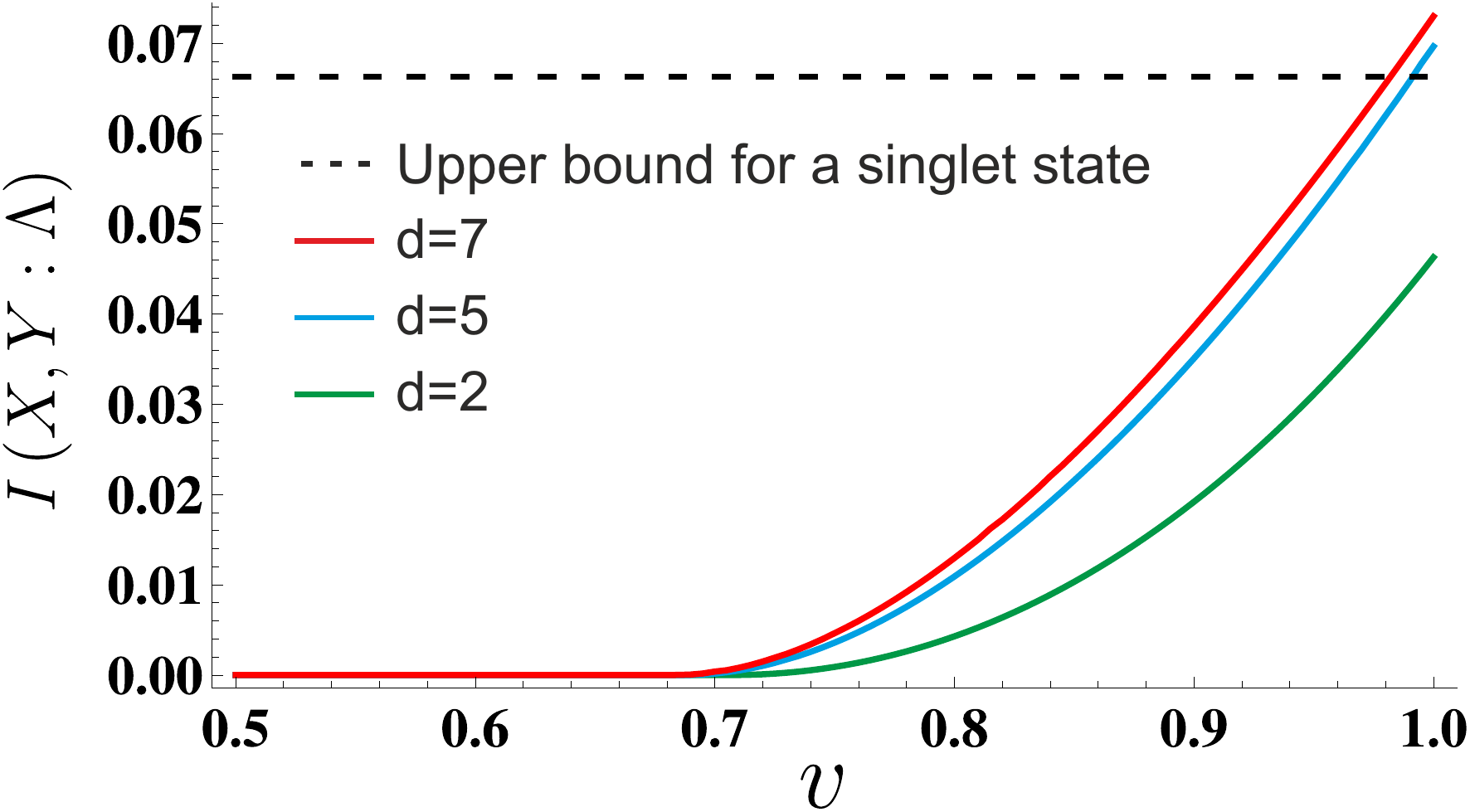}
\caption{Upper bound for $I(X,Y:\lambda)$ computed as a function of the visibility $V$ for $d=2,5,7$ (green, blue and red curves, respectively). The black dashed curve correspond to the upper bound  $I(X,Y:\Lambda) \approx 0.0663$ obtained in \cite{Hall2011} for singlet states. The solid curves correspond to $v p^{Q}_{max} +(1-v)p_{\text{W}}$ were $p^{Q}_{max}$ was obtained by maximizing the quantum violation of $I_{d}$ over pure states and projective measurements.}
\label{fig:CGLMPplot}
\end{figure}

\section{Bilocality scenario}

In LHV models for multipartite Bell scenarios, it is usually assumed that the same hidden variable is shared among all the parties. That is, a Bell inequality violation rules out any shared LHV. However, in quantum information protocols it is often the case that different parties receive particles produced by independent sources, e.g.~in quantum networks \cite{Cavalcanti2011quantum,Fritz2012,Chaves2014,Chaves2014b,Chaves2014information}. It is then natural to focus on LHV models which reproduce the independence structure of the sources. That is, each hidden variable can only be shared between parties receiving particles from the same source. Such models are weaker than general LHV models, i.e.~they form a subset of all the models where the hidden variables can be shared arbitrarily among the parties.

A particular case is an entanglement swapping scenario \cite{Zukowski1993} involving three parties $A$, $B$ and $C$ which receive entangled states from two independent sources. The DAG of \figref{fig:models}f shows an LHV model with independent variables for this scenario. The assumption that the sources are independent, $p(\lambda_1,\lambda_2)=p(\lambda_1)p(\lambda_2)$, is known as bilocality \cite{Branciard2010,Branciard2012}. With this assumption, in analogy with the usual LHV decomposition \eqref{LHV}, the correlations for this scenario must fulfil
\begin{eqnarray}
\label{eq.bilocLHV}
p(a,b,c\vert x,z)=\sum_{\lambda_1,\lambda_2} & & p(\lambda_1)p(\lambda_2) \\ \nonumber
& & p(a\vert x,\lambda_1)p(b\vert\lambda_1,\lambda_2) p(c\vert z,\lambda_2) .
\end{eqnarray}
Note that the set of bilocal correlations is non-convex because of the nonlinearity of the bilocality assumption. This makes the set extremely difficult characterize~\cite{Branciard2010,Branciard2012,Fritz2012,FritzChaves2013,Chaves2012,Chaves2013a}. In the following, we introduce a measure of relaxation of bilocality, and we show that, despite the non-convex nature of the measure, it can nevertheless be computed by means of a linear program.

For fixed numbers $m_x$, $m_z$ and $o_a$, $o_b$, $o_c$ of the input $x$, $z$ and output $a$, $b$, $c$ values, there is a finite number $n=o_a^{m_x} o_b o_c^{m_z}$ of deterministic strategies. We can label the deterministic strategies for $a$ by symbols $\bar{\alpha}=\alpha_0,\ldots,\alpha_{m_x}$ where $\alpha_x$ is the value of $a$ when the input is $x$. Similarly, we label the functions for $b$ by $\beta$ and for $c$ by $\bar{\gamma} = \gamma_0,\ldots,\gamma_{m_z}$. Thus, the distribution over the deterministic strategies can be identified with an $n$-dimensional vector ${\bf q}$, analogous to the case in the main text for usual LHV models. The vector ${\bf q}$ then has components $q_{\bar{\alpha},\beta,\bar{\gamma}}$.
Defining the marginals
\begin{equation}
\begin{split}
q^{ac}_{\bar{\alpha},\bar{\gamma}} & = \sum_\beta q_{\bar{\alpha},\beta,\bar{\gamma}} \\
q^{a}_{\bar{\alpha}} = \sum_{\beta,\bar{\gamma}} q_{\bar{\alpha},\beta,\bar{\gamma}} , & \hspace{0.5cm} q^{c}_{\bar{\gamma}} = \sum_{\beta,\bar{\alpha}} q_{\bar{\alpha},\beta,\bar{\gamma}} ,
\end{split}
\end{equation}
the bilocality assumption is equivalent to the requirement
\begin{equation}
q^{ac}_{\bar{\alpha},\bar{\gamma}} = q^a_{\bar{\alpha}}q^c_{\bar{\gamma}} .
\end{equation}
In analogy with the measure \eqref{meas_corr} of measurement dependence, the degree of non-bilocality can be measured by how much the distribution over the LHVs fail to comply with this criterion. We define the measure of non-bilocality as
\begin{equation}
\mathcal{M}_{\text{BL}}= \sum_{\bar{\alpha},\bar{\gamma}} \vert q^{ac}_{\bar{\alpha},\bar{\gamma}} - q^{a}_{\bar{\alpha}}q^{c}_{\bar{\gamma}} \vert.
\end{equation}
Clearly $\mathcal{M}_{\text{BL}}=0$ if and only if the bilocality constraint is fulfilled.

The non-bilocality measure is quadratic in the distribution over the the deterministic strategies. Thus, it is not obvious that linear programming will be helpful in computing $\mathcal{M}_{\text{BL}}$ or that the computation can be made efficient. However, we notice that, for given observed correlations, there are restrictions on the marginals $q^a_{\bar{\alpha}}$ and $q^c_{\bar{\gamma}}$ imposed by the observed distribution $p(a,b,c|x,z)$ because of the constraint \eqref{eq.bilocLHV} that the LHV must reproduce the observations. This constraint can be written
\begin{equation}
\label{eq.bilocLHVq}
p(a,b,c|x,z) = \sum_{\bar{\alpha},\beta,\bar{\gamma}} \delta_{a,\alpha_x}\delta_{b,\beta}\delta_{c,\gamma_z} q_{\bar{\alpha},\beta,\bar{\gamma}} .
\end{equation}
Depending on the observed distribution, there may be no or just a few free parameters $\nu$ which determine $q^{a}_{\bar{\alpha}} = f_{\bar{\alpha}}(\nu)$. We can then rewrite $\mathcal{M}_{\text{BL}}$ as
\begin{equation}
\mathcal{M}_{\text{BL}}(\nu) = \sum_{\bar{\alpha},\bar{\gamma}} \vert q^{ac}_{\bar{\alpha},\bar{\gamma}} - f_{\bar{\alpha}}(\nu) q^{c}_{\bar{\gamma}} \vert .
\end{equation}
For fixed $\nu$ the measure $\mathcal{M}_{\text{BL}}(\nu)$ is linear and its minimum can be found via a linear program, as we now show.

As previously, the first step is to write $\mathcal{M}_{\text{BL}}(\nu)$ as an $\ell_1$-norm. For a given value of $\nu$, we can write
\begin{align}
\mathcal{M}_{\text{BL}}(\nu) & = \sum_{\bar{\alpha},\bar{\gamma}} \vert \sum_{\beta} q_{\bar{\alpha},\beta,\bar{\gamma}} - f_{\bar{\alpha}}(\nu) \sum_{\bar{\alpha}',\beta} q_{\bar{\alpha}',\beta,\bar{\gamma}} \vert \\
& = \sum_{\bar{\alpha},\bar{\gamma}} \vert \sum_{\bar{\alpha}'\beta'\bar{\gamma}'} M^\nu_{\bar{\alpha}\bar{\gamma},\bar{\alpha}'\beta'\bar{\gamma}'} q_{\bar{\alpha}'\beta'\bar{\gamma}'} \vert  \\
& =  \| M^\nu \mathbf{q} \|_{\ell_1} ,
\end{align}
where $M^\nu$ is a matrix of dimension $l \times n$, with $l=o_a^{m_x} o_c^{m_z}$ and entries $M^\nu_{\bar{\alpha}\bar{\gamma},\bar{\alpha}'\beta'\bar{\gamma}'} = \delta_{\bar{\alpha},\bar{\alpha}'} \delta_{\bar{\gamma},\bar{\gamma}'} - f_{\bar{\alpha}}(\nu) \delta_{\bar{\gamma},\bar{\gamma}'}$ (where $\delta_{\bar{\alpha},\bar{\alpha}'} = \delta_{\alpha_0,\alpha_0'}\cdots \delta_{\alpha_{o_x},\alpha_{o_x}'}$ etc.). Minimisation of $\mathcal{M}_{\text{BL}}(\nu)$ for given, observed correlations $p(a,b,c|x,z)$ is then equivalent to
\begin{equation}
\label{eq.bilocL1min}
\begin{split}
\underset{\mathbf{q}\in\mathbb{R}^n}{\minimize} \; & \| M^\nu \mathbf{q} \|_1 		\\
\st & A \mathbf{q} = \mathbf{p} 			\\
    &  \langle \1_n,  \mathbf{q} \rangle = 1 	\\
   &  \mathbf{q} \geq \mathbf{0}_n ,
\end{split}
\end{equation}
where $\mathbf{p}$ is the $k$-dimensionsal vector representing the observed correlations, with $k=o_ao_bo_cm_xm_z$, and $A$ is a $k \times n$ matrix which encodes the constraint \eqref{eq.bilocLHVq} that the LHV must reproduce the observations. The entries of $A$ are $A_{abcxz,\bar{\alpha}\beta\bar{\gamma}} = \delta_{a,\alpha_x}\delta_{b,\beta}\delta_{c,\gamma_z}$. From Theorem~\ref{thm:1}, the minimisation \eqref{eq.bilocL1min} is equivalent to the linear program
\begin{equation}
\begin{split}
\underset{\mathbf{t} \in \mathbb{R}^l}{\minimize} \,\, & \langle \1_l, \mathbf{t} \rangle 		\\
\st & - \mathbf{t} \leq M^\nu \mathbf{q} \leq \mathbf{t} , 		\\
 & A\mathbf{q} = \mathbf{p}, 	\\
 & \langle \1_n, \mathbf{q} \rangle = 1, 		\\
 & \mathbf{q} \geq \mathbf{0}_n
\end{split}
\end{equation}
Thus, minimising $\mathcal{M}_{\text{BL}}(\nu)$  for fixed $\nu$ is indeed a linear program. To find the minimum of the measure $\mathcal{M}_{\text{BL}}$ we must minimise also over $\nu$ and hence we have an optimisation over a linear program. In order to verify non-bilocality of a given distribution we need to check that the minimum over $\nu$ is non-zero, or equivalently that the minimum of $\mathcal{M}_{\text{BL}}(\nu)$ is non-zero for all values of $\nu$ in the allowed range. On the other hand, if we find a value of $\nu$ such that $\mathcal{M}_{\text{BL}}(\nu) = 0$ this is sufficient to show that the distribution is bilocal (and as a by-product we get an explicit bilocal decomposition).

\subsection{Bilocality with binary inputs}

\begin{figure}[!t]
\includegraphics[width=1.0\columnwidth]{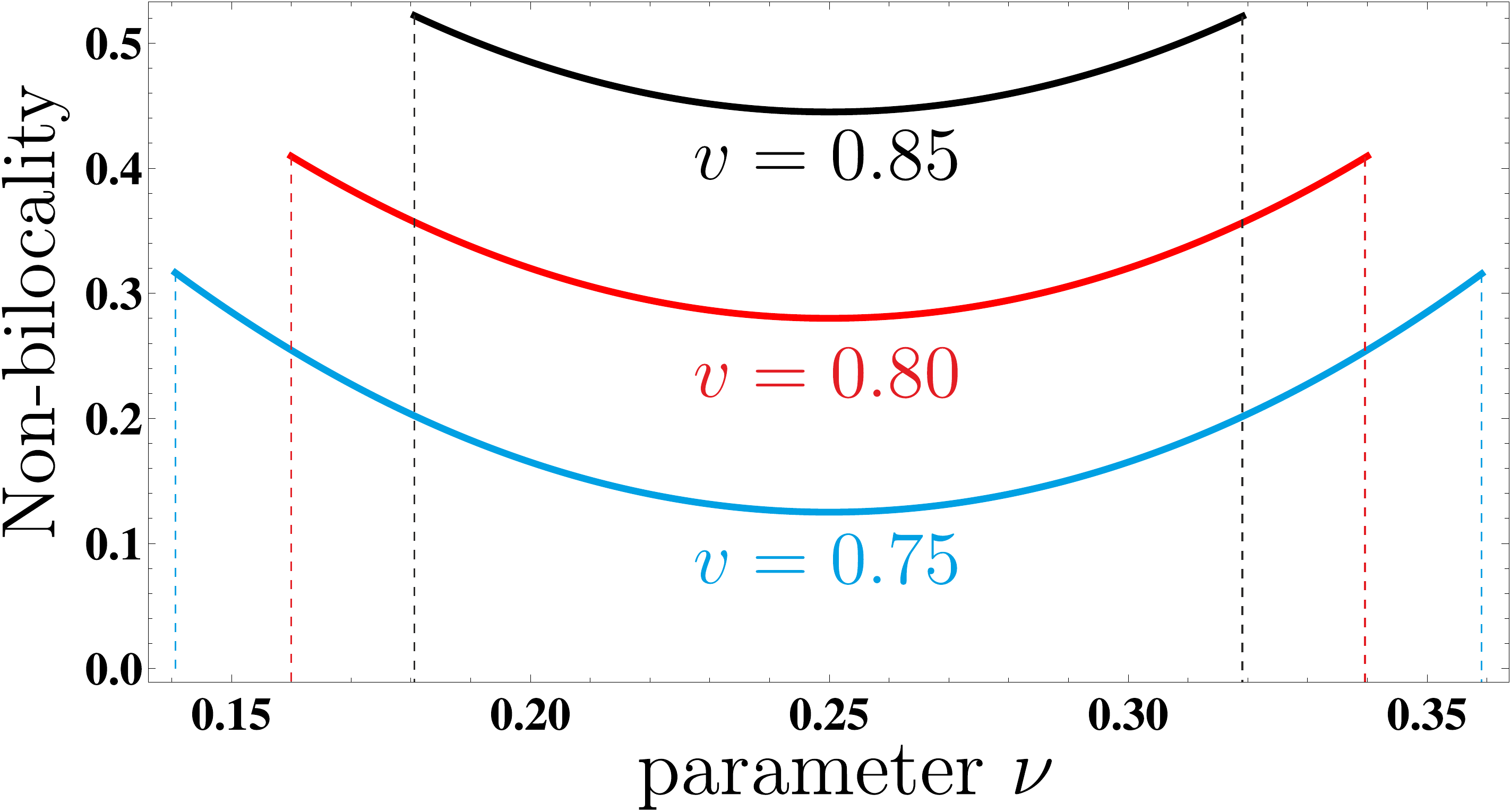}
\caption{$\mathcal{M}_{\text{BL}}(\nu)$ as a function of $\nu$ for three different values of the visibility ($v=0.75$ (blue curve), $v=0.80$ (red curve) and $v=0.85$ (black curve)). The dashed lines correspond to the minimum and maximum values of the parameter $\nu$ that are compatible with the probability distribution. We observe that for the specific distribution considered the minimum of $\mathcal{M}_{\text{BL}}(\nu)$ is achieved for $\nu=1/4$.}
\label{bilocal2}
\end{figure}

To illustrate our framework, and to compare with previous results, we now consider the case where the inputs and ouputs of $A$ and $C$ are all dichotomic ($o_a=o_c=m_x=m_z=2$), and the output of $B$ takes four values ($o_b=4$) that we decompose as $b=(b_0,b_1)$ where $b_0$, $b_1$ are bits. Furthermore, we consider the distribution \cite{Branciard2010,Branciard2012}
\begin{equation}
\label{pv_dist}
p_v\left(  a,b,c | x,z \right)  = v^2 p\left(  a,b,c | x,z \right) + (1-v^2)\frac{1}{16}
\end{equation}
with
\begin{equation}
p\left(  a,b,c | x,z \right)  = \frac{1}{16}\left( 1+ (-1)^{a+c}\frac{(-1)^{b_0}+(-1)^{x+z+b_1}}{2} \right)
\end{equation}
This distribution can be obtained by using shared Werner states with visibility $v$, that is $\varrho=v\ket{\Psi^{-}}\bra{\Psi^{-}}+(1-v)\mathbb{I}/4$, on which Alice and Charlie perform measurements given by $A_0=C_0=\frac{1}{\sqrt{2}}(Z+X)$ and $A_1=C_1=\frac{1}{\sqrt{2}}(Z-X)$, while Bob measures in the Bell basis assigning $b_0b_1=00,01,10,11$ to $\ket{\Phi^{+}}$, $\ket{\Phi^{-}}$, $\ket{\Psi^{+}}$ and $\ket{\Psi^{-}}$. As shown in \cite{Branciard2010,Branciard2012} this distribution is non-bilocal. Taking the marginal of \eqref{eq.bilocLHVq} gives $p(a|x) = \sum_{\bar{\alpha}} \delta_{a,\alpha_x} q_{\bar{\alpha}}$, explicitly for the distribution \eqref{pv_dist}
\begin{equation}
\begin{split}
p(a=0\vert x=0) & = q^{a}_{0,0} +q^{a}_{0,1} = \frac{1}{2} \\
p(a=0\vert x=1) & = q^{a}_{0,0} +q^{a}_{1,0} = \frac{1}{2} \\
p(a=1\vert x=0) & = q^{a}_{1,0} +q^{a}_{1,1} = \frac{1}{2} \\
p(a=1\vert x=1) & = q^{a}_{0,1} +q^{a}_{1,1} = \frac{1}{2} .
\end{split}
\end{equation}
This implies that $q^{a}_{0,0}=q^{a}_{1,1}$ and $q^{a}_{1,0}=q^{a}_{0,1}=1/2-q^{a}_{0,0}$ and thus we have a single free parameter $\nu = q^{a}_{0,0}$. The parameter is further constrained by the full distribution $p(a,b,c\vert x,z)$. To determine its range we run the following two linear programs
\begin{equation}
\begin{split}
\minimize \,\, & \langle \mathbf{c} , \mathbf{q} \rangle 	\\
\st & A \mathbf{q} = \mathbf{p},		\\
& \mathbf{q} \geq \mathbf{0}_n ,
\end{split}
\end{equation}
and
\begin{equation}
\begin{split}
\maximize \,\, & \langle \mathbf{c} , \mathbf{q} \rangle 	\\
\st & A \mathbf{q} = \mathbf{p},		\\
& \mathbf{q} \geq \mathbf{0}_n ,
\end{split}
\end{equation}
where $\langle\mathbf{c},\mathbf{q}\rangle = q^a_{0,0}$. These two linear programs define a range $\nu_{min} \leq \nu \leq \nu_{max}$. In some particular cases $\nu_{max}=\nu_{min}$, in which case the minimisation over $\nu$ is superfluous and the minimum of $\mathcal{M}_{\text{BL}}$ is directly given by a linear program and is thus analytical. However, in general these bounds are different. For the distribution \eqref{pv_dist} with $v=1$, we have $\nu_{max}=\nu_{min}=1/4$, while for $v=0.8$ we have $\nu_{min}=0.16$ and $\nu=0.34$. In general what we observe is that for any $v$, the minimum $\mathcal{M}_{\text{BL}}(\nu)$ occurs at $\nu=1/4$. This is illustrated \figref{bilocal2}.

In \figref{bilocal1} we show how the minimum of $\mathcal{M}_{\text{BL}}$ depends on the visibility. We also show the value of the bilocality quantity $\mathcal{B}=\sqrt{|I|}+\sqrt{|J|}$ given in \cite{Branciard2010,Branciard2012}, where
\begin{equation}
\begin{split}
I & =\frac{1}{4} \displaystyle \sum_{x,z=0}^1 \langle A_{x}B^{0}C_{z} \rangle, \\
J & =\frac{1}{4} \displaystyle \sum_{x,z=0}^1 (-1)^{x+z} \langle A_{x}B^{1}C_{z} \rangle ,
\end{split}
\end{equation}
and
\begin{equation}
\langle A_{x}B^{y}C_{z} \rangle= \displaystyle \sum_{a,b_0,b_1,c} (-1)^{a+b_{y}+c} p(a,b_{0},b_{1},c|x,z).
\end{equation}
In \cite{Branciard2010,Branciard2012} it was shown that bilocality implies $\mathcal{B} \leq 1$. For the distribution \eqref{pv_dist} on the other hand, $I=J=\frac{1}{2} v^2 $ and therefore $\mathcal{B}=\sqrt{2}v$. From the numerical results in \figref{bilocal1} one can easily fit the data and find $\min \mathcal{M}_{\text{BL}}=\mathcal{B}^2-1$. Thus the violation of the bilocality corresponds exactly to how much bilocality must be relaxed to reproduce the observed distribution.

\begin{figure}[!t]
\includegraphics[width=1.0\columnwidth]{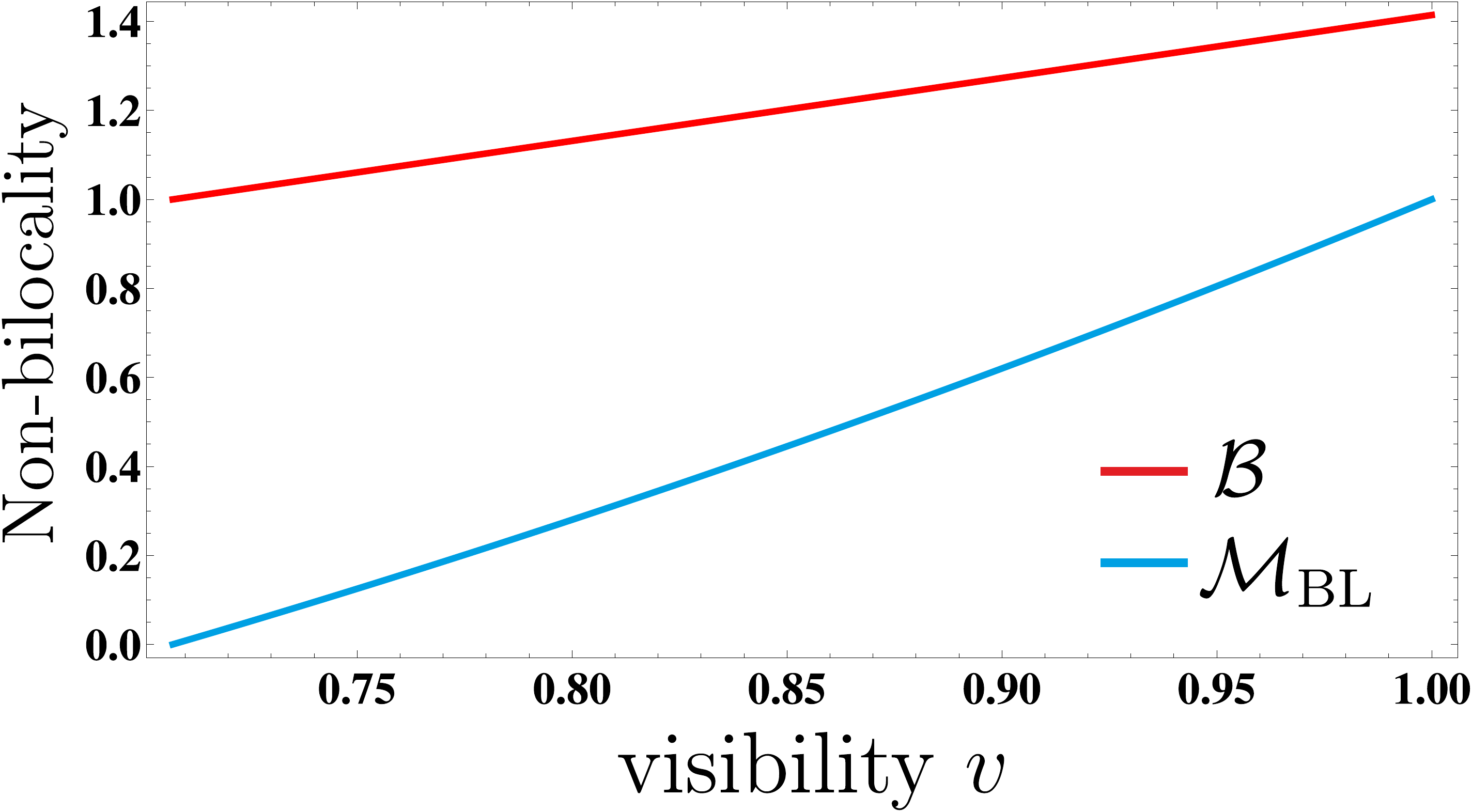}
\caption{Minimum of the non-bilocality measure $\mathcal{M}_{\text{BL}}$ vs.~visibility $v$ (blue). We also show the bilocality quantity $\mathcal{B}$ (red). Our measure can be understood as the amount of correlation between the sources required to simulate the observed correlations.}
\label{bilocal1}
\end{figure}

\subsection{Bilocality with ternary inputs}

To sketch how the linear framework could be used in more general bilocality scenarios we consider the case where Alice and Charlie can perform three different measurements. We again consider the case of trivial marginals $p(a\vert x)=1/2$. This imposes the following constraints on $q^a_{\alpha_0,\alpha_1,\alpha_2}$
\begin{equation}
\begin{split}
p(a=0\vert x=0) & = q^{a}_{0,0,0} +q^{a}_{0,0,1} +q^{a}_{0,1,0} +q^{a}_{0,1,1} = \frac{1}{2} \\
p(a=0\vert x=1) & = q^{a}_{0,0,0} +q^{a}_{0,0,1} +q^{a}_{1,0,0} +q^{a}_{1,0,1} = \frac{1}{2} \\
p(a=0\vert x=2) & = q^{a}_{0,0,0} +q^{a}_{0,1,0} +q^{a}_{1,0,0} +q^{a}_{1,1,0} = \frac{1}{2} \\
p(a=1\vert x=0) & = q^{a}_{1,0,0} +q^{a}_{1,0,1} +q^{a}_{1,1,0} +q^{a}_{1,1,1} = \frac{1}{2} \\
p(a=1\vert x=1) & = q^{a}_{0,1,0} +q^{a}_{0,1,1} +q^{a}_{1,1,0} +q^{a}_{1,1,1} = \frac{1}{2} \\
p(a=1\vert x=2) & = q^{a}_{0,0,1} +q^{a}_{0,1,1} +q^{a}_{1,0,1} +q^{a}_{1,1,1} = \frac{1}{2} ,
\end{split}
\end{equation}
which implies that
\begin{equation}
\begin{split}
q^{a}_{011} &= \frac{1}{2} - q^{a}_{000} - q^{a}_{001} - q^{a}_{010}\\
q^{a}_{101} &= \frac{1}{2} - q^{a}_{000} - q^{a}_{001} - q^{a}_{100}\\
q^{a}_{110} &= \frac{1}{2} - q^{a}_{000} - q^{a}_{010} - q^{a}_{100}\\
q^{a}_{111} &=-\frac{1}{2} + 2 q^{a}_{000} + q^{a}_{001} + q^{a}_{010} + q^{a}_{100}.
\end{split}
\end{equation}
This means that we now have four free parameters $\nu = (q^{a}_{000},q^{a}_{001}, q^{a}_{010}, q^{a}_{100})$ . To linearize $\mathcal{M}_{\text{BL}}$ in this case, we need to optimize over these four variables.

In practice, given a certain distribution $p(a,b,c\vert x,z)$, we first fix a certain value for $q^{a}_{000}=c_0$ in the range  $q^{min}_{000} \leq q^{a}_{000} \leq q^{max}_{000}$. We then solve a linear program to find the bounds for the next free parameter $q^{min}_{001} \leq q^{a}_{001} \leq q^{max}_{001} $ but now imposing also the constraint that $q^{a}_{000}=c_0$. Fixed $q^{a}_{000}=c_0$ and $q^{a}_{001}=c_1$ we look for the bounds of the next free parameter $q^{min}_{010} \leq q^{a}_{010} \leq q^{max}_{010}$. We now run the linear program for the remaining free parameter in the range $q^{min}_{100} \leq q^{a}_{100} \leq q^{max}_{100} $ determined by the probability distribution and the constraints $q^{a}_{000}=c_0$, $q^{a}_{001}=c_1$, $q^{a}_{010}=c_2$.

For a sufficiently good discretization of these continuous free parameters, we can be quite confident about the non-bilocality of the distribution if we find no values for which $\mathcal{M}_{\text{BL}}\neq 0$. On the other hand, if we find any values for the free parameters such that $\mathcal{M}_{\text{BL}} = 0$, then we can immediately conclude that the distribution is bilocal. To illustrate this we have tested the distribution obtained with two maximally entangled states $\ket{\Psi^-}$ when Alice and Charlie measure the three observables $X,Y,Z$ while Bob measures in the Bell basis. It is possible to show that this distribution is bilocal by setting  $q^{a}_{000}=0$ and $q^{a}_{001}=q^{a}_{010}=q^{a}_{100}=1/4$.

\end{document}